\documentclass[a4paper, 10pt, journal]{ieeeconf}      %
\IEEEoverridecommandlockouts %
\renewcommand{\QED}{\QEDopen}

\usepackage{amsmath} %
\usepackage{thmtools}

\newtheorem{theorem}{Theorem}[section]
\newtheorem{proposition}[theorem]{Proposition}

\newtheorem{lemma}[theorem]{Lemma}
\newtheorem{corollary}[theorem]{Corollary}

\newtheorem{remark}[theorem]{Remark}
\newtheorem{assumption}[theorem]{Assumption}

\newtheorem{example}[theorem]{Example}

\newtheorem{problem}[theorem]{Problem}
\newtheorem{open problem}[theorem]{Open problem}
\newtheorem{definition}[theorem]{Definition}

\usepackage{amsfonts} %
\usepackage{amssymb} %
\usepackage{amsopn} %
\usepackage{ifthen}
\usepackage{suffix} %

\newcommand \mc {\mathcal}

\DeclareMathAlphabet{\mbb}{U}{bbold}{m}{n} 
\newcommand \mb [1] {\ifthenelse{\equal{#1}{0}}{\mbb{0}}{\ifthenelse{\equal{#1}{1}}{\mbb{1}}{\mathbb{#1}}}} %

\newcommand \RR {\mb R}

\newcommand \NN {\mb N}

\newcommand{\on}{\operatorname}

\newcommand \inv {^{-1}}
\newcommand \T {^\top}
\newcommand \Tinv {^{-\top}}
\newcommand \invT {^{-\top}}
\renewcommand \H {^\mathrm{*}}
\newcommand \comp {^{\mathrm{c}}}

\WithSuffix \newcommand \inv* {\mathstrut^{-1}}
\WithSuffix \newcommand \T* {\mathstrut^\top}
\WithSuffix \newcommand \Tinv* {\mathstrut^{-\top}}
\WithSuffix \newcommand \invT* {\mathstrut^{-\top}}
\WithSuffix \newcommand \H* {\mathstrut^\mathrm{*}}
\WithSuffix \newcommand \comp* {\mathstrut^{\mathrm{c}}}

\let\tilde\widetilde

\let\hat\widehat

\renewcommand \subset {\subseteq}

\newcommand \casesif {\text{if }}

\newcommand \casesand {\text{ and }}

\newcommand \im {\on{im}}
\newcommand \diag [1] {\on{diag}}
\newcommand {\cl} [1] {\operatorname{cl}(#1)}
\newcommand {\conv} [1] {\operatorname{conv}(#1)}
\newcommand {\inter} [1] {\operatorname{int}(#1)}

\newcommand {\myvphantom} [1] {\let \\ \empty \vphantom{#1}}
\newcommand {\set} [2] {%
\left\{%
\myvphantom{#1#2}%
\right.%
#1%
\left|%
\myvphantom{#1#2}%
\right.#2%
\left.%
\myvphantom{#1#2}%
\right\}%
}

\renewcommand \d {\mathrm{d}}

\newcommand{\pdd}[2]{\frac{\partial#1}{\partial#2}}
\newcommand{\tpdd}[2]{\tfrac{\partial#1}{\partial#2}}

\newcommand \Ie {\textit{I.e.}}
\def \ie {\textit{i.e.}}
\newcommand \eg {\textit{e.g.}}

\renewcommand{\l}{\left}
\renewcommand{\r}{\right}

\usepackage{pdfpages} %
\usepackage{xcolor} %
\usepackage{graphicx} %
\usepackage{enumerate} %
\def\mytitle{DC power grids with constant-power loads\textemdash Part I: A full characterization of power flow feasibility, long-term voltage stability and their correspondence}
\def\mythanks{This work is supported by NWO (Netherlands Organisation for Scientific Research) project `Energy management strategies for interconnected smart microgrids' within the DST-NWO Joint Research Program on Smart Grids.}

\def\myauthors{Mark Jeeninga, Claudio De Persis and Arjan van der Schaft}
\def\myaffiliation{University of Groningen, 9747AG Groningen, The Netherlands (e-mail: \{m.jeeninga, c.de.persis, a.j.van.der.schaft\}@rug.nl)}

\title{\mytitle}

\author{\myauthors%
	\thanks{\mythanks}%
	\thanks{\myaffiliation}%
}

\begin{document}

\maketitle
\thispagestyle{empty}
\pagestyle{empty}

\begin{abstract}
In this two-part paper we develop a unifying framework for the analysis of the feasibility of the power flow equations for DC power grids with constant-power loads.

In Part I of this paper we present a detailed introduction to the problem of power flow feasibility of such power grids, and the associated problem of selecting a desirable operating point which satisfies the power flow equations.
We introduce and identify all long-term voltage semi-stable operating points, and show that there exists a one-to-one correspondence between such operating points and the constant power demands for which the power flow equations are feasible.
Such operating points can be found by solving an initial value problem, and a parametrization of these operating points is also obtained.
In addition, we give a full characterization of the set of all feasible power demands, and give a novel proof for the convexity of this set.
Moreover, we present a necessary and sufficient LMI condition for the feasibility of a vector of power demands under small perturbation, which extends a necessary condition in the literature.
\end{abstract}
\section{Introduction}\label{section:introduction}
A classical problem in the study of power grid stability is the long-term voltage stability problem. %
The problem concerns the long-term (in)stability of a power grid due to limitations in the transportation of power from sources to loads.
These limits in power transportation are due to a combination of generation limits, load limits and/or limits due to the network (infra)structure.
The transportation of power from sources to loads is known as \emph{power flow} or \emph{load flow}, and is captured in the \emph{power flow equations} (or also, {load flow equations}).
Consequently, the power flow equations implicitly describe the limitations of power flow in a power grid.

Another motivation for the study of power flow are phenomena such as voltage drop, voltage collapse and power outages. Such phenomena may occur when transportation limits are exceeded and the power flow equations cannot be satisfied. 
Loosely speaking, control schemes which are designed to satisfy the power flow equations in the long-term time scale may display unintended behavior when the the power flow equations cannot be satisfied. A possible consequence is that critical components may reach their operational limits, start to fail, and cause a chain reaction of more failures.
Satisfaction of the power flow equations is therefore crucial
to guarantee (long-term) safe operation of the power grid.

Long-term voltage instability is a load-driven phenomenon \cite{cutsem2008voltage}, and different load characteristics may be considered for analysis\textemdash see, \eg, \cite{eminoglu2005power}.
For a load characteristic with given parameters, we refer to solutions of the power flow equations as \emph{operating points} of the power grid.
The power flow equations are \emph{feasible} if at least one operating point exists. 
In general, the power flow equations are nonlinear, and no operating points may exist.
Likewise, multiple operating points may exist, while a single operating point should be selected.

For practical power grids there are several distinct properties to select an operating point.
First, it is desirable that an operating point is \emph{long-term voltage stable}, meaning that all voltage magnitudes of an operating point decrease if any load demand increases \cite{cutsem2008voltage}. 
This is to say that the Jacobian of the voltage magnitudes at the loads as a function of the power demand has negative elements \cite{simpson2016voltage}.
Second, it is desirable that the selected operating point is the solution to the power flow equations that minimizes the total power dissipated in the lines at steady state.
A third property is that the operating point is a \emph{high-voltage solution}, meaning that the selected operating point element-wise dominates all other operating point that satisfy the power flow equation.

It is \textit{a priori} not clear if, or under which conditions, these types of operating points coincide, or give rise to a unique operating point.
For specific types of power grids it has been shown that these types of operating point coincide ``almost surely'', and that a sufficient condition exists for the uniqueness of the long-term voltage stable operating point \cite{matveev2020tool}. %
A similar result for a general power grid is not available in the literature, but several sufficient conditions for are known \cite{dorfler2013synchronization,simpson2016voltage,bolognani2015existence}.

There have been several publications on the problem of long-term voltage stability, and in particular on the feasibility of power flow equations. Here we list a few of them.\\
The paper \cite{bolognani2015existence} considers a generic AC power grid with a single source, whereas
the paper \cite{simpson2016voltage} considers a lossless AC power grid. Both \cite{simpson2016voltage} and \cite{bolognani2015existence} give a sufficient condition for feasibility of the reactive power flow equations, and use fixed-point methods to conclude the existence of an operating point. Estimates of the operating point are also given. In addition, \cite{simpson2016voltage} shows that this operating point is the high-voltage solution, and that it is long-term voltage stable.\\
The paper \cite{barabanov2016} considers a general power transportation system at steady state and proves a necessary conditions for feasibility of a power demand, which is also sufficient in certain cases.\\
The paper \cite{matveev2020tool} presents an algorithm to determine if the power flow equations of a DC power grid are feasible, and shows that there exists a high-voltage solution, which is ``almost surely'' long-term voltage stable. Conditions for the long-term voltage stable operating point to be unique are given.\\
The paper \cite{dymarsky2014convexity} proves that the set of feasible solutions to the power flow equations for DC power grids is convex, as follows from the study of convexity of the non-homogeneous numerical range of a generalized quadratic form.

In this two-part paper we focus on the power flow of DC power grids with constant-power loads, and assume there are no limits on voltage potentials and line currents.
While many types of power grids are studied in the literature, %
DC power grids with constant-power loads are among the simplest type of power grid where feasibility of the power flow equations is nontrivial.

It is noted that there are several types of power grids for which the power flow equations are equivalent to or well-approximated by the power flow of DC power grids with constant-power loads.
The paper \cite{dorfler2013novel} shows how the active power flow problem for a lossless AC power grid may be approximated by a DC power flow grid. 
The papers \cite{simpson2015solvability,simpson2016voltage} show a similar result for the reactive power flow problem for a lossless AC power grid. %
See also \cite{bolognani2015existence,matveev2020tool} for other examples.
Even though the literature provides handles to study power flow, the interplay between the different results is not clear, and an over-arching analysis is missing. The main motivation of this paper is to bridge these gaps in the literature for DC power grids, and to develop a unified framework for the analysis of DC power flow with constant-power loads.

\subsection*{Contribution}
This paper is split into two parts. 
Part~I of this paper presents a geometric framework to analyze the feasibility of the DC power flow equations with constant-power loads.
This framework is extended in Part~II to unify and generalize the main contributions of the previously mentioned publications in the context of DC power flow.
The novelty of our approach is that we combine results in matrix theory, convex analysis, and initial value problems to analyze DC power flow. 
In contrast to other approaches, we do not explicitly rely on fixed point analysis or iterative methods.

The main objective of this twin paper is to analyze the set, denoted by $\mc F$, of constant power demands for which the power flow equations are feasible.
We would like to emphasize that these constant power demands are not sign-restricted. \Ie, the power demand at a load is allowed to be negative, in which case the load provides power to the grid. %
We let $\mc D$ denote the set of long-term voltage stable operating points. 
We refer to the vectors in $\cl{\mc D}$, the closure of $\mc D$, as \emph{long-term voltage semi-stable} operating points.
In Part~I of this paper we develop a geometric framework for DC power flow feasibility for constant-power loads.
The main contributions of Part~I are as follows.
\newcommand{\mainresult}[1]{\textbf{M\ref{main results:#1}}}
\begin{enumerate}[\hspace{5pt}\bf M1.]
\item We give a parametrization of $\mc D$, its closure and its boundary, which establishes a constructive method to describe the long-term voltage (semi-)stable operating points (Theorem \ref{theorem:parametrization of D}). \label{main results:parametrization of D}
\item For each vector of power demand that lies on the boundary of $\mc F$ there exists a unique corresponding operating point which solves the power flow equations. Moveover, these operating points form the boundary of $\mc D$ (Corollary~\ref{corollary:one-to-one boundary}).%
\label{main results:one-to-one boundary}%
\item There is a one-to-one correspondence between the feasible power demands $\mc F$ and the long-term voltage semi-stable operating points $\cl{\mc D}$. This means that if the power flow equations are feasible, then there exists a unique long-term voltage semi-stable operating point that solves the power flow equations.
This operating point can be found by solving an initial value problem (Theorem~\ref{theorem:one-to-one correspondence}). \label{main results:one-to-one correspondence}%
\item We give a novel and insightful proof for the fact that the set $\mc F$ is closed and convex. Consequently, $\mc F$ is the intersection of all supporting half-spaces of $\mc F$. We describe all such half-spaces, which gives a complete geometric characterization of $\mc F$ (Theorem~\ref{theorem:convexity of F}).\label{main results:convexity of F}%
\item We prove %
a necessary and sufficient LMI condition for the feasibility of the power flow equations, %
and a necessary and sufficient LMI condition for the feasibility of the power flow equations under small perturbations of the power demands (Theorem \ref{theorem:necessary and sufficient condition}).\label{main results:necessary and sufficient condition}
\end{enumerate}

Part II of this paper continues the approach, and recovers and extends several results of the previously mentioned publications for DC power grids with constant-power loads. The main contributions of Part~II are as follows.
\begin{enumerate}[\hspace{5pt}\bf M1.]\setcounter{enumi}{5}
\item We give an alternative parametrization of $\mc D$, its closure and its boundary.\label{main results:second parametrization of D}
\item We give two parametrizations of $\partial\mc F$, the boundary of the set of feasible power demands.\label{main results:boundary of F}
\item We refine the results of \mainresult{boundary of F} and \mainresult{necessary and sufficient condition} for nonnegative power demands, which are cheaper to compute.\label{main results:nonnegative power demands}
\item We prove that any vector of power demands that is element-wise dominated by a feasible vector of power demands is also feasible.\label{main results:power demand domination}
\item We present two novel sufficient conditions for the feasibility of the power flow equations which generalize the sufficient conditions in \cite{simpson2016voltage} and \cite{bolognani2015existence}, and show how these conditions are related.\label{main results:sufficient conditions} %
\item We show that the long-term voltage stable operating point is a strict high-voltage solution.
Consequently, the operating points associated to a feasible power demand which are either long-term voltage stable, a high-voltage solution, or dissipation-minimizing, are one and the same.\label{main results:desirable operating point}
\end{enumerate}

It is important to explain how these results are related to the existing literature, and in which regard these results are, to the best of the authors' knowledge, novel.

Regarding \mainresult{one-to-one correspondence}, it was shown in \cite{matveev2020tool} that if the power flow equations are feasible, then there ``almost surely'' exists an operating point which is long-term voltage stable, and that it is the unique long-term voltage stable operating point if all power demands are positive, or all are negative. %
By studying long-term voltage \textit{semi}-stable operating points, we show %
that for each feasible vector of power demands there always exists a unique long-term semi-stable operating point. \\ %
Regarding \mainresult{convexity of F}, the convexity of $\mc F$ was already shown in \cite{dymarsky2014convexity} (see also \cite{Dymarsky_2019}), and follows from an analysis of the convexity of the numerical range of non-homogeneous quadratic maps. Our approach to prove convexity is different from and less general than the one proposed by \cite{dymarsky2014convexity}, and is a byproduct of the proof of \mainresult{one-to-one correspondence}, the one-to-one correspondence between $\mc F$ and $\cl{\mc D}$. We believe our proof for convexity to be simpler.\\ %
Regarding \mainresult{necessary and sufficient condition}, our contribution is a necessary and sufficient condition for the feasibility of power demands under small perturbations. In \cite{barabanov2016} a similar condition was shown to be sufficient for a more general system with constant-power loads at steady-state. It was shown in \cite{barabanov2016} to also be necessary whenever $\mc F$ is closed convex, as is the case here.\\
Regarding \mainresult{desirable operating point}, it was shown in \cite{matveev2020tool} that, if the power flow equations are feasible, then there exists a high-voltage solution, \ie, an operating point that element-wise dominates all other operating points which satisfy the power flow equations, and that this operating point is ``almost surely'' long-term voltage stable. 
We show that the element-wise domination is strict, and that this operating point always 
coincides with the unique long-term voltage semi-stable operating point.
This shows the algorithm proposed in \cite{matveev2020tool} converges to the unique long-term voltage semi-stable operating point wherever the power flow equations are feasible.\\

\subsection*{Organization of Part I}

In Section~\ref{section:problem section} we formulate the DC power flow equations, discuss the problem of their feasibility, and define different types of desirable operating points.
We give a detailed introduction to this feasibility problem and its difficulties.

In Section~\ref{section:problem analysis} we develop a geometric framework to analyze the DC power flow equations. 
The main objective of Section~\ref{section:problem analysis} is to prove that there is a one-to-one correspondence between $\mc F$ and $\mc D$, and to give a method to compute the desired operating point (\mainresult{one-to-one correspondence}). In addition, we prove that $\mc F$ is convex and present a full geometric characterization of $\mc F$ as an intersection of half-spaces (\mainresult{convexity of F}). 
To establish this, we present a parametrization of $\mc D$ (\mainresult{parametrization of D}), and prove that there is a one-to-one correspondence between the boundary of $\mc D$ and the boundary of the convex hull of $\mc F$, which is a prelude to proving that there is a one-to-one correspondence between the boundary of $\mc D$ and the boundary of $\mc F$ (\mainresult{one-to-one boundary}). The section is concluded by presenting a necessary and sufficient LMI condition for the feasibility of a vector of power demands, and a similar condition for feasibility under small perturbation (\mainresult{necessary and sufficient condition}).

Section~\ref{section:conclusion} %
concludes Part I of the paper.

\subsection*{Notation and matrix definitions}
For a vector $x = \begin{pmatrix}
x_1 & \cdots & x_k
\end{pmatrix}\T$ we denote 
\begin{align*}
[x]:=\on{diag} (x_1,\dots,x_k). 
\end{align*}
We let $\mb 1$ and $\mb 0$ denote the all-ones and all-zeros vector, respectively, and let $I$ denote the identity matrix.
We let their dimensions follow from their context.
All vector and matrix inequalities are taken to be element-wise. 
We write $x \lneqq y$ if $x\le y$ and $x\neq y$.
We let $\|x\|_p$ denote the $p$-norm of $x\in\RR^k$.

We define $\boldsymbol n := \{1,\dots,n\}$.
All matrices are square $n\times n$ matrices, unless stated otherwise.
The submatrix of a matrix $A$ with rows and columns indexed by $\alpha,\beta\subset\boldsymbol n$, respectively, is denoted by $A_{[\alpha,\beta]}$. The same notation $v_{[\alpha]}$ is used for subvectors of a vector $v$.
We let $\alpha\comp$ denote the set-theoretic complement of $\alpha$ with respect to $\boldsymbol n$. 
For a set $S$, the notation $\inter S$, $\cl S$, $\partial S$ and $\conv S$ is used for the interior, closure, boundary and convex hull of $S$, respectively.

We list some classical definitions from matrix theory.
\begin{definition}[\cite{fiedler1986special}, Ch. 5]\label{definition:Z-matrix}
A matrix $A$ is a \emph{Z-matrix} if $A_{ij}\le 0$ for all $i\neq j$.
\end{definition}
\begin{definition}[\cite{fiedler1986special}, Thm. 5.3]\label{definition:M-matrix}
A Z-matrix is an \emph{M-matrix} if all its eigenvalues have nonnegative real part.
\end{definition}
\begin{definition}[\cite{fiedler1986special}, pp. 71]\label{definition:irreducible matrix}
A matrix $A$ is \emph{irreducible} if for every nonempty set $\alpha\subsetneqq \boldsymbol n$ we have $A_{[\alpha,\alpha\comp]}\neq 0$.
\end{definition}

\section{The DC power flow problem}\label{section:problem section}
In this section we formulate the DC power flow equations and explore their feasibility. 
We study DC power grids that consist of nodes (buses), which are either loads or sources, and are interconnected by lines.
Source nodes are voltage controlled buses which provide power to the power grid, and represent generators such as power plants.
Load nodes are voltage controlled buses which generically extract power from the power grid. %
The power flow equations describe the power balance at the load nodes.
We are interested in the existence of a solution to the power flow equations in the long-term time scale, and therefore study DC power grid at steady-state. Note that, at steady-state, the line dynamics do not contribute to this power balance. %
We therefore model the power grid as a resistive circuit. We refer to \cite{schaft2010characterization,schaft2019flow} for a detailed discussion on resistive circuits.

We proceed with the modeling of DC power grids with purely resistive lines at steady state.
For the sake of simpli- city we do not consider operational limits on lines, currents or voltage potentials in the power grid.
We consider a DC power grid with $n$ load nodes and $m$ source nodes.
We write $i\sim j$ if there exists a line between node $i$ and node $j$, and $i\not\sim j$ otherwise.
The conductance of the line between node $i$ and node $j$ is denoted by $w_{ij}=w_{ji}$, which is a positive real number.
The Kirchhoff matrix $Y \in \RR^{(n+m)\times(n+m)}$ describes both the topology and the conductances of the lines in the power grid. It is defined by
\begin{align*}
Y_{ij} := \begin{cases}
\sum_{k\sim i}w_{ik} & \casesif i=j\\
-w_{ij} & \casesif i \neq j \casesand i\sim j\\
0 & \casesif i \neq j \casesand i\not\sim j
\end{cases}.
\end{align*}
Note that $Y$ is a symmetric Z-matrix of which all rows and columns sum to zero. 
Hence $Y\mb 1=\mb 0$, and thus the matrix $Y$ is singular.
We assume that the nodes and lines form a connected graph. This implies that $\mb 1$ spans the kernel of $Y$ \cite{schaft2010characterization}, and that all principal submatrices of $Y$ are invertible \cite{schaft2019flow}. %
One may verify that for vectors $x$, $z$ we have
\begin{align}\label{eqn:positive definiteness Y}
x\T Y z = \sum_{i,j:~i\sim j} w_{ij} (x_i-x_j)(z_i-z_j).
\end{align}
Since $w_{ij}>0$ whenever $i\sim j$, \eqref{eqn:positive definiteness Y} implies $x\T Yx \ge 0$, and hence $Y$ is positive semi-definite.
Consequently, all principal submatrices of $Y$ are positive definite.

We partition $Y$ according to whether nodes are loads ($L$) or sources ($S$):
\begin{align}
Y = \begin{pmatrix}\label{eqn:partition of Y}
Y_{LL} & Y_{LS} \\ Y_{SL} & Y_{SS}
\end{pmatrix}.
\end{align}
The matrices $Y_{LL}$ and $Y_{SS}$ are positive definite, as they are principal submatrices of $Y$.
Following the same partition, let
\begin{align*}
V = \begin{pmatrix}
V_L \\ V_S
\end{pmatrix}\in\RR^{n + m}
\end{align*}
denote the vector of voltage potentials at the nodes. 
All voltage potentials are assumed to be positive (\ie, $V>\mb 0$).

We let $\mc I \in \RR^{n+ m}$ denote the electric current injected into the power grid at the nodes. The power that a node provides to the power grid is given by the vector 
\begin{align*}%
P = [V]\mc I \in \RR^{n + m}.
\end{align*}
Kirchhoff's laws together with Ohm's law state that
\begin{align}\label{eqn:KCL}
\begin{pmatrix}
\mc I_L \\ \mc I_S
\end{pmatrix} = \begin{pmatrix}
Y_{LL} & Y_{LS} \\ Y_{SL} & Y_{SS}
\end{pmatrix}\begin{pmatrix}
V_L \\ V_S
\end{pmatrix},
\end{align}
and therefore
\begin{align}\label{eqn:power at the nodes}
\begin{pmatrix}
P_L \\ P_S
\end{pmatrix} = \begin{bmatrix}
\begin{pmatrix}
V_L \\ V_S
\end{pmatrix}
\end{bmatrix}\begin{pmatrix}
Y_{LL} & Y_{LS} \\ Y_{SL} & Y_{SS}
\end{pmatrix}\begin{pmatrix}
V_L \\ V_S
\end{pmatrix}.
\end{align}
The \emph{total dissipated power} in the lines is derived in \cite{schaft2010characterization} as
\begin{align*}
R(V_L,V_S) := V\T Y V \ge 0.
\end{align*}

\subsection{Feasible constant power demands}\label{subsection:feasible constant power demands}
Throughout this paper we consider $V_L$ as a variable of the system, whereas the Kirchhoff matrix $Y$ and the voltages at the sources $V_S$ are fixed. 
We therefore write $\mc I_L = \mc I_L(V_L)$ and $P_L = P_L(V_L)$.

\begin{definition}
We define the \emph{source-injected currents} by 
\begin{align}\label{eqn:source-injected currents}
\mc I_L^*:=-Y_{LS}V_S = -\mc I_L(\mb 0),
\end{align}
which correspond to the currents injected into the loads by the sources when $V_L=\mb 0$.
\end{definition}
\begin{proposition}[{\cite[Prop. 3.6]{schaft2010characterization}}]
The \emph{open-circuit voltages} $V_L^*$ are the unique voltage potentials at the loads so that $\mc I_L(V_L) = \mb 0$, given by
\begin{align}\label{eqn:open-circuit voltages}
V_L^* := -{Y_{LL}}\inv Y_{LS}V_S = {Y_{LL}}\inv \mc I_L^*.
\end{align}
\end{proposition}
\begin{lemma}\label{lemma:positive open-circuit voltages}
The source-injected currents are nonnegative and not all zero (\ie, $\mc I_L^*\gneqq \mb 0$) and the open-circuit voltages $V_L^*$ are positive (\ie, $V_L^*>\mb 0$).
\end{lemma}
\begin{proof}
Let $\alpha_1,\dots,\alpha_k\subset\boldsymbol n$ index the connected components of the graph formed by loads and the lines between them. 
The matrices $(Y_{LL})_{[\alpha_i,\alpha_i]}$ are therefore irreducible, while $(Y_{LL})_{[\alpha_i,\alpha_j]} = 0$ for $i\neq j$ \cite{fiedler1986special}. %
Note that $(Y_{LL})_{[\alpha_i,\alpha_i]}$ is a principal submatrix of a $Y$, and is therefore positive definite, and that its off-diagonal elements are nonpositive. This implies that $(Y_{LL})_{[\alpha_i,\alpha_i]}$ is an irreducible nonsingular M-matrix, and therefore has a positive inverse \cite[Thm. 5.12]{fiedler1986special}.\linebreak
It follows that ${Y_{LL}}\inv$ is (permutation similar to) a block diagonal matrix with positive diagonal blocks.
The matrix $-Y_{LS}$ is a nonnegative matrix and the vector $V_S$ is positive. Hence, $-Y_{LS}V_S = \mc I_L^*$ is nonnegative. Since the graph of $Y$ is connected, there exists a line between the connected component represented by $\alpha_i$ and a source node. This implies that $(\mc I_L^*)_{[\alpha_i]}$ does not equal $\mb 0$, and thus $(\mc I_L^*)_{[\alpha_i]}\gneqq \mb 0$. We conclude that the vector
\begin{align*}
(V_L^*)_{[\alpha_i]} = ({Y_{LL}}\inv)_{[\alpha_i,\alpha_i]}(\mc I_L^*)_{[\alpha_i]} = {(Y_{LL})_{[\alpha_i,\alpha_i]}}\inv(\mc I_L^*)_{[\alpha_i]}
\end{align*}
is a positive vector, since it is the product of a positive matrix and a nonnegative nonzero vector.
This observation holds for all $i$, and thus $V_L^*>\mb 0$ and $\mc I_L^*\gneqq \mb 0$.
\end{proof}

If we substitute \eqref{eqn:open-circuit voltages} into \eqref{eqn:power at the nodes}, we obtain the equation
\begin{align}\label{eqn:power at the nodes 2}
P_L(V_L) = [V_L]Y_{LL} (V_L - V_L^*).
\end{align}
The open-circuit voltages are the unique voltage potentials at the loads which satisfy $V_L>\mb 0$ and $P_L(V_L) = \mb 0$.

In this paper we consider \emph{constant-power loads}, which is to say that each load demands a fixed quantity of power from the power grid. We let $P_c\in\RR^n$ denote the vector of constant power demands.
Note that we do not impose any sign restrictions on $P_c$ and that, in principle, load nodes could also demand negative power, in which case the loads provide constant power to the grid.
By equating the power demand with the power injection at the load nodes, we obtain the power balance
\begin{align}\label{eqn:power balance}
P_L(V_L) + P_c = \mb 0.
\end{align}
Note that power demand and power injection have opposite signs, and that indeed the vectors in \eqref{eqn:power balance} should be summed.
The substitution of \eqref{eqn:power at the nodes 2} in \eqref{eqn:power balance} yields the \emph{DC power flow equation for constant-power loads}:
\begin{align}\label{eqn:dc power flow equation}
[V_L]Y_{LL} (V_L - V_L^*) + P_c = \mb 0.
\end{align}
\begin{definition}\label{definition:feasible}
Given $Y$ and $V_S$, we say that the power flow equations \eqref{eqn:dc power flow equation} are \emph{feasible} for a vector of power demands $P_c$ if there exists a vector of voltage potentials $V_L$ %
which satisfies \eqref{eqn:dc power flow equation}. We say that $V_L$ is as an \emph{operating point} associated to $P_c$ if $V_L$ satisfies \eqref{eqn:dc power flow equation} for $P_c$. 
\end{definition}

Recall that thoughout Definition~\ref{definition:feasible} we require that ${V_S>\mb 0}$ and $V_L>\mb 0$.
\begin{definition}
We say that a vector of power demands $P_c$ is \emph{feasible} if \eqref{eqn:dc power flow equation} is feasible for $P_c$.
The set of \emph{feasible power demands} is given by
\begin{align*}
\mc F :=& \set{P_c}{\text{Eq. \eqref{eqn:dc power flow equation} is feasible for $P_c$}}.
\end{align*}
\end{definition}

Since \eqref{eqn:dc power flow equation} is a quadratic equation in $V_L$, the existence of a solution to \eqref{eqn:dc power flow equation} for a given $P_c$ is not guaranteed. 
Furthermore, \eqref{eqn:dc power flow equation} may have multiple solutions, and multiple operating points for a single $P_c$ may exist.

Our goal is to characterize all constant power demands $P_c$ such that \eqref{eqn:dc power flow equation} is feasible, which is precisely the set $\mc F$. 
This is formalized in the following problem statement.
\begin{problem}\label{problem:dc power flow problem}
Consider a DC power grid with Kirchhoff matrix $Y$ and voltage potentials at the sources $V_S>\mb 0$. 
Let $V_L^*$ be given by \eqref{eqn:open-circuit voltages}.
For which power demands at the loads $P_c\in \RR^n$ does there exist an operating point $V_L>\mb 0$ of voltage potentials at the loads which satisfies \eqref{eqn:dc power flow equation}?
\end{problem}

We make a few observations. %
First, note that Problem~\ref{problem:dc power flow problem} is not affected by the lines between sources, since the matrix $Y_{SS}$ does not appear in \eqref{eqn:dc power flow equation}. More specifically, since \eqref{eqn:dc power flow equation} only depends on $Y_{LL}$, $V_L^*$, $P_L$ and $V_L$, it follows that Problem~\ref{problem:dc power flow problem} only depends on $Y_{LL}$ and $V_L^*$, or on $Y_{LL}$ and $\mc I_L^*$ by \eqref{eqn:open-circuit voltages}.
Second, if the graph formed by loads and the lines between them is not connected, then the matrix $Y_{LL}$ is permutation similar to a block diagonal matrix with multiple blocks, as was observed in the proof of Lemma~\ref{lemma:positive open-circuit voltages}.
It follows that \eqref{eqn:dc power flow equation} can be analyzed for each block separately. Hence, without loss of generality, we make the assumption that the graph formed by the loads and the lines between them is connected.
\begin{assumption}\label{assumption:loads connected}
The load nodes and the lines between loads form a connected graph, or equivalently by \cite[Thm. 3.6.a]{fiedler1986special}, the matrix $Y_{LL}$ is irreducible.
\end{assumption}

\subsection{Desirable operating points}
For a feasible power demand $P_c$ there may be multiple operating points which satisfy \eqref{eqn:dc power flow equation}. We are generically interested in the following two criteria to determine a desirable operating point. %

\paragraph{Long-term voltage stable operating points}
First, we desire that the selected operating point is such that a small increase in a single power demand leads to a small decrease in all voltage potentials \cite{cutsem2008voltage, simpson2016voltage}. %
Note that each vector $V_L$ of voltage potentials at the loads is associated by \eqref{eqn:dc power flow equation} to a vector of constant power demands $P_c$, %
given by
\begin{align}\label{eqn:definition of f}
P_c(V_L) = [V_L]Y_{LL}(V_L^*-V_L).
\end{align}
We remark that \eqref{eqn:definition of f} should be interpreted as the vector of constant power demand which are satisfied by the vector $V_L$ of voltage potentials at steady state, and that constant power demands do not depend on the voltage potentials at the loads.
By virtue of \eqref{eqn:definition of f} the Jacobian of $P_c$ at an operating point $\tilde V_L$ is well-defined.
We use the following definition from \cite{simpson2016voltage}.\footnote{In \cite{simpson2016voltage} this property is referred to as local voltage stability, whereas we prefer the term long-term voltage stability.}
\begin{definition}\label{definition:long-term voltage stable}
An operating point $\tilde V_L$ associated to $\tilde P_c$ is \emph{long-term voltage stable} if the Jacobian of $P_c(V_L)$ at $\tilde V_L$ is nonsingular, %
and its inverse is a matrix with negative elements (\ie,\footnote{The equality $\pdd {P_c}{V_L}(\tilde V_L)\inv = \pdd {V_L}{P_c}(\tilde P_c)$ holds locally and follows from the Inverse Function Theorem, see \eg\ \cite{rudin1964principles}.} $\pdd {P_c}{V_L}(\tilde V_L)\inv = \pdd {V_L}{P_c}(\tilde P_c) < 0$). The set of all long-term voltage stable operating points is defined by
\begin{align*}
\mc D := \set{V_L}{%
\begin{minipage}{162pt}\text{$\exists \tilde P_c$ such that $V_L$ is a long-term voltage} \\\text{stable operating point associated to $\tilde P_c$}\end{minipage}
}.
\end{align*}
\end{definition}\vspace{4pt}

We remark that there are many equivalent definitions and names for long-term voltage stability.
For example, \cite{simpson2016voltage} refers to operating points described by Definition~\ref{definition:long-term voltage stable} as locally voltage stable, whereas \cite{matveev2020tool} uses the term voltage-regularity.
We refer to Remark~\ref{remark:equivalent definitions of long-term voltage stability} for a more detailed discussion.

Similar to Definition~\ref{definition:long-term voltage stable} we define the notion of long-term voltage semi-stability:
\begin{definition}\label{definition:long-term voltage semi-stable}
An operating point $\tilde V_L$ associated to $\tilde P_c$ is \emph{long-term voltage semi-stable} if for every $\varepsilon>0$ there exists a long-term voltage stable operating point $\widehat V_L$ associated to some $\widehat P_c$ such that $\|\tilde V_L- \widehat V_L\|_2<\varepsilon$. Consequently, the set of all long-term voltage semi-stable operating points equals $\cl{\mc D}$, the closure of $\mc D$.
\end{definition}\vspace{2pt} 
We emphasize that it is \textit{a priori} not clear that each feasible $P_c$ has a long-term voltage semi-stable operating point.

\paragraph{Dissipation-minimizing operating points}
Second, it is desirable that an operating point $V_L$ associated to $P_c$ minimizes $R(V_L,V_S)$, the total power dissipated in the lines.
\begin{definition}\label{definition:dissipation-minimizing operating point}
Given $P_c$, an operating point $\tilde V_L$ associated to $\tilde P_c$ is \emph{dissipation-minimizing} if for all operating points $V_L$ associated to $\tilde P_c$ we have $R(\tilde V_L,V_S)\le R(V_L,V_S)$.
\end{definition}
For such operating points the following proposition applies.
\begin{proposition}\label{proposition:minimal total dissipation}
Given $\tilde P_c\in\mc F$, an operating point $\tilde V_L$ associated to $\tilde P_c$ is \emph{dissipation-minimizing} if and only if it maximizes $V_L\T \mc I_L^*$ among all operating points $V_L$ associated to $\tilde P_c$, where $\mc I_L^*$ is the quantity defined in \eqref{eqn:source-injected currents}.
\end{proposition}
\begin{proof}
We define $\mc S$ to be the set of all $V_L>\mb 0$ which satisfies \eqref{eqn:dc power flow equation}. 
We write $R(V_L,V_S) = V\T Y V$ in terms of the partitioning in \eqref{eqn:partition of Y} and substitute \eqref{eqn:open-circuit voltages} and \eqref{eqn:source-injected currents}. This results in %
\begin{align}
&R(V_L,V_S) \nonumber \\ %
&= V_L\T (Y_{LL} V_L + Y_{LS}V_S) + V_S\T (Y_{SS} V_S + Y_{SL}V_L)\nonumber\\
& = V_L\T Y_{LL} (V_L - V_L^*) + V_S\T Y_{SS} V_S - V_L\T \mc I_L^*.\label{eqn:total dissipated power reformulation}
\end{align}
For any $V_L\in \mc S$, multiplying \eqref{eqn:dc power flow equation} by $\mb 1\T$ shows that
\begin{align}\label{eqn:total power demand}
V_L\T Y_{LL} (V_L - V_L^*) = - \mb 1 \T P_c.
\end{align}
We are interested in minimizing $R(V_L,V_S)$ over all $V_L\in \mc S$.
By using \eqref{eqn:total dissipated power reformulation} and substituting \eqref{eqn:total power demand}, for all $V_L\in \mc S$ we have
\begin{align}\label{eqn:total dissipated power for operationg points}
R(V_L,V_S) = \underbrace{-\mb 1\T P_c + V_S\T Y_{SS} V_S}_{\text{fixed}} -V_L\T \mc I_L^*.
\end{align}
Only the last term in \eqref{eqn:total dissipated power for operationg points} depends on $V_L$, whereas the other terms are fixed.
Thus, minimizing $R(V_L,V_S)$ over all $V_L\in \mc S$ is equivalent to maximizing $V_L\T \mc I_L^*$ over all $V_L\in \mc S$.
\end{proof}

It was shown in \cite{matveev2020tool} that for each feasible vector of power demands there
exists an operating point associated to a given $P_c$ which element-wise dominates all other operating points associated to $P_c$. Such an operating point is referred to as the \emph{high-voltage solution} to \eqref{eqn:dc power flow equation} (see also \cite{simpson2016voltage}). We formalize this notion by the following definition.
\begin{definition}\label{definition:high-voltage solution}
An operating point $\tilde V_L$ associated to the power demands $P_c$ is a \emph{high-voltage solution} if $\tilde V_L \ge V_L$ for all $V_L$ associated to $P_c$, and is a \emph{strict high-voltage solution} if $\tilde V_L > V_L$ for all $V_L\neq\tilde V_L$ associated to $P_c$.
\end{definition}

It follows from Proposition~\ref{proposition:minimal total dissipation} that a high-voltage solution is always dissipation-minimizing.
\begin{corollary}\label{corollary:high-voltage}
If the operating point $\tilde V_L$ associated to the power demands $P_c$ is a high-voltage solution, then $\tilde V_L$ is dissipation-minimizing. Moreover, if $\tilde V_L$ is a strict high-voltage solution, then $\tilde V_L$ is the unique dissipation-minimizing operating point.
\end{corollary}
\begin{proof}
Let $V_L$ be an operating point associated to $P_c$.
If $\tilde V_L-V_L\ge \mb 0$ then $(\tilde V_L-V_L)\T \mc I_L^* \ge 0$ since $\mc I_L^*\ge \mb 0$ by Lemma~\ref{proposition:minimal total dissipation}, and hence $\tilde V_L\T \mc I_L^*\ge V_L\T \mc I_L^* $. If $\tilde V_L-V_L> \mb 0$ then $(\tilde V_L-V_L)\T \mc I_L^* > 0$ since $\mc I_L^*\ge \mb 0$ and $\mc I_L^*\neq\mb 0$, and hence $\tilde V_L\T \mc I_L^*> V_L\T \mc I_L^*$. The result follows from Proposition~\ref{proposition:minimal total dissipation}.%
\end{proof}

Note that Definitions~\ref{definition:long-term voltage stable} and \ref{definition:long-term voltage semi-stable} describe local properties of an operating point, while Definitions~\ref{definition:dissipation-minimizing operating point} and \ref{definition:high-voltage solution} are global properties concerning all operating points associated to $P_c$.
It is \textit{a priori} not clear how dissipation-minimizing operating points and long-term voltage semi-stable operating points are related, nor is it clear when a feasible vector of power demands has a (possibly unique) long-term voltage semi-stable operating point or when it has a strict high-voltage solution.
Some partial answers to these questions are known\textemdash see \cite{matveev2020tool,simpson2016voltage}\textemdash but a full characterization is lacking.
This fundamental question is answered in this paper.
Indeed, %
in Part I of this paper we show that for each feasible vector of power demands there exists a unique long-term voltage semi-stable operating point associated to $P_c$ (\mainresult{one-to-one correspondence}). In Part II we show that this operating point is a strict high-voltage solution (\mainresult{desirable operating point}), which proves that the aforementioned notions coincide due to Corollary~\ref{corollary:high-voltage}. %
\subsection{Academic examples of DC power flow with constant-power loads}
In this section we explore the intricacies of Problem~\ref{problem:dc power flow problem} by considering two simple examples. We will focus on building some intuition for the sets $\mc F$ and $\mc D$.
We first consider the simplest case of a DC power grid with constant-power loads.

\begin{example}[Single load case]\label{example:one node example}
Consider a DC power grid with a single load and a single source (\ie, $n=m=1$), as depicted in Figure~\ref{figure:one load node schematic}. The corresponding graph of $P_L(V_L)$ is given in Figure~\ref{figure:one load node parabola}. Figure~\ref{figure:one load node operating points} depicts the relation between $P_c$ and $V_L$. In this example we let $V_S>\mb 0$ be arbitrary and recall that the open-circuit voltages $V_L^*>\mb 0$ are defined by \eqref{eqn:open-circuit voltages}. Since $n=1$, it follows that \eqref{eqn:dc power flow equation} is scalar-valued. By taking \eqref{eqn:dc power flow equation} and completing the squares we find
\begin{align*}
Y_{LL} (V_L - \tfrac 1 2 V_L^*)^2 - \tfrac 1 4 Y_{LL} (V_L^*)^2 + P_c = 0.
\end{align*}
Since $n=1$, $Y_{LL}$ is a positive scalar, and it follows that 
\begin{align}\label{eqn:one node example:operating point}
V_L = \tfrac 1 2 V_L^* \pm \sqrt{{Y_{LL}}\inv( \tfrac 1 4 Y_{LL} (V_L^*)^2 - P_c)}.
\end{align}
We see that \eqref{eqn:dc power flow equation} has a real solution for $n=1$ if and only if 
\begin{align}\label{eqn:one node example:power demand inequality}
{P_c \le \tfrac 1 4 Y_{LL} (V_L^*)^2}.
\end{align}
The set of all feasible power demands is therefore given by
\begin{align}\label{eqn:one node example:F}
\mc F = \set{P_c}{P_c \le \tfrac 1 4 Y_{LL} (V_L^*)^2}.
\end{align}
If equality holds in \eqref{eqn:one node example:power demand inequality}, it follows from \eqref{eqn:one node example:operating point} that there is precisely one operating point, given by $V_L = \tfrac 1 2 V_L^*$. In the case that \eqref{eqn:one node example:power demand inequality} is strict, we see that the positive branch of \eqref{eqn:one node example:operating point} leads to a higher voltage potential at the load, and minimizes $R(V_L,V_S)$. Hence, the positive branch of \eqref{eqn:one node example:operating point} is the high-voltage solution.
In addition, the positive branch decreases when $P_c$ increases, and so it is also the long-term voltage stable solution. We have $\mc D = \set{V_L}{V_L>\tfrac 1 2 V_L^*}$.\hfill\QED
\end{example}

\begin{figure}%
  \centering
  \includegraphics[width=.5\linewidth]{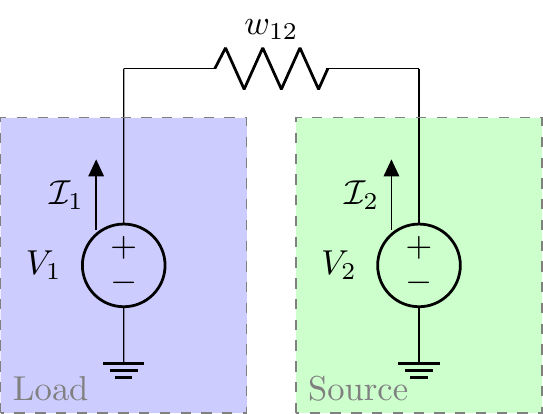}
  \caption{\label{figure:one load node schematic}
A schematic depiction of a power grid with a single load node and a single source node ($n=m=1$), where $w_{12}$ is the conductance of the line between the nodes.}
\end{figure}
\begin{figure}%
\centering
\includegraphics[width=.9\linewidth]{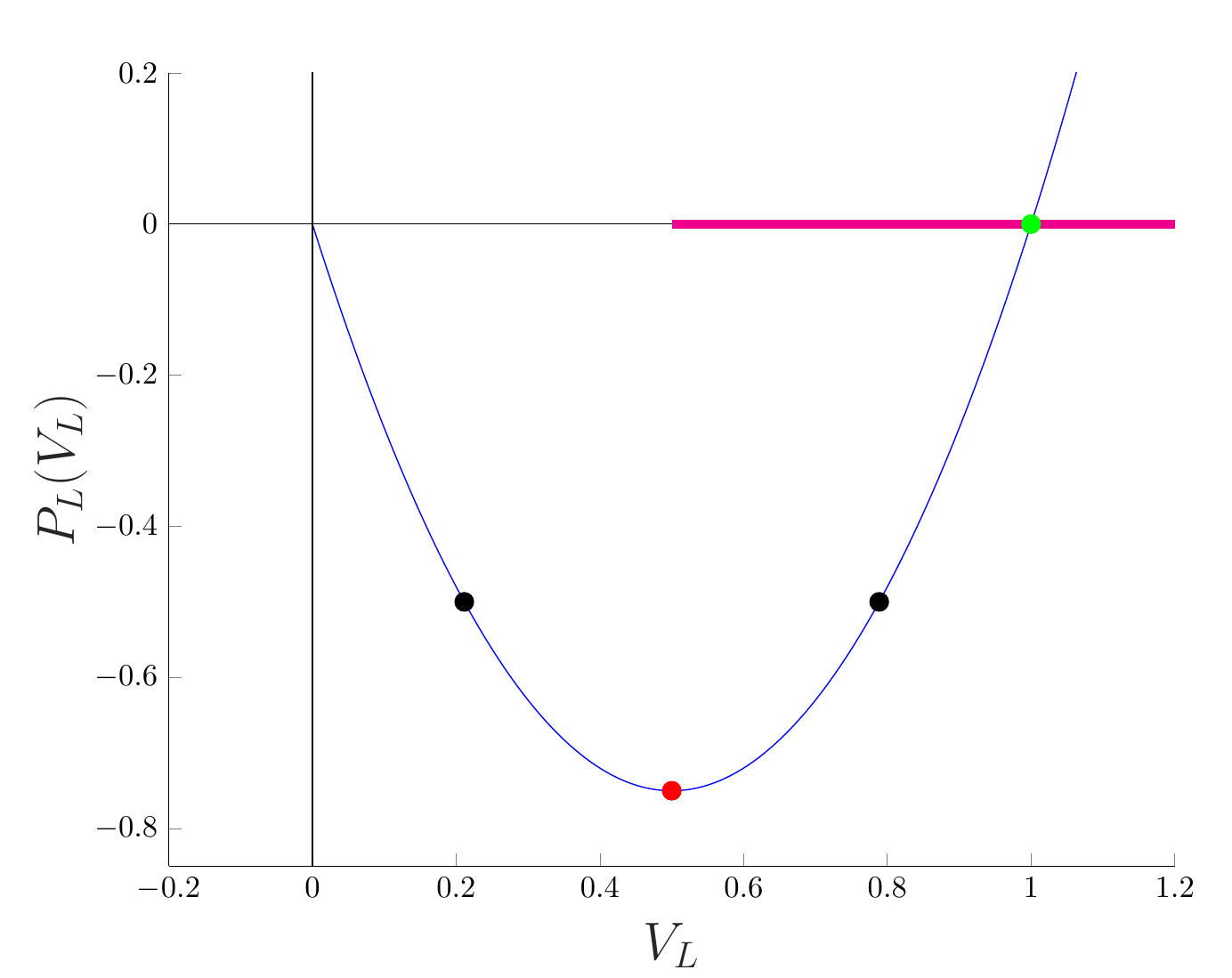}
\caption{\label{figure:one load node parabola}
A plot of $P_L(V_L)$ against $V_L$ for the power grid in Figure~\ref{figure:one load node schematic} with $w_{12} = 3~\Omega\inv$ and $V_S = V_2 = 1~\on{V}$. The red point indicates the voltage such that the power that the grid transports is maximized. The thick purple half-line corresponds to $\mc D$, the set of all long-term voltage stable operating points.\vspace{-10pt}}
\end{figure}
\begin{figure}[bt]
\centering
\includegraphics[width=.9\linewidth]{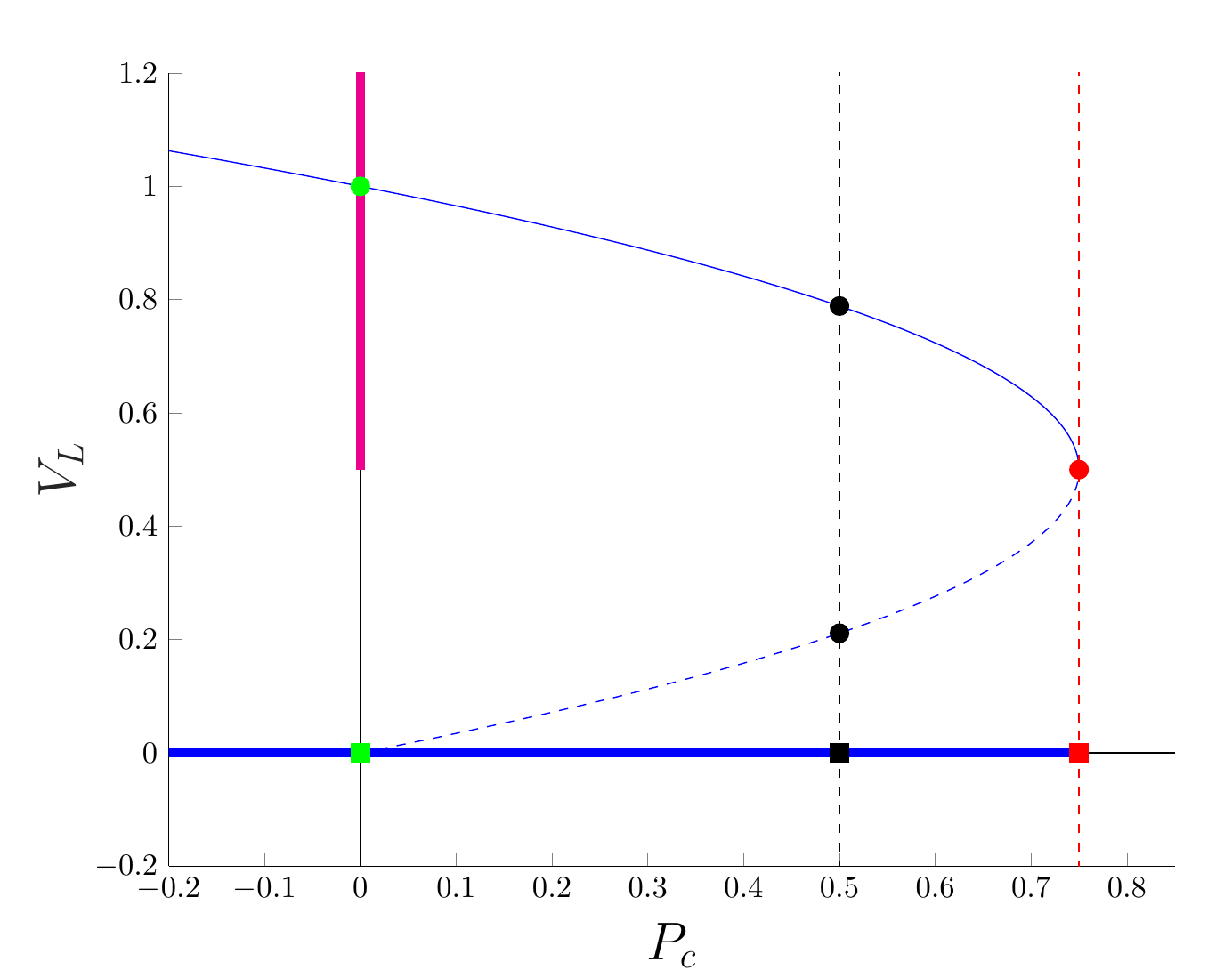}
\caption{\label{figure:one load node operating points}
A plot of $V_L$ against $P_c$ for the power grid in Figure~\ref{figure:one load node schematic} with $w_{12} = 3~\Omega\inv$ and $V_S = V_2 = 1~\on{V}$. The thick blue half-line corresponds to all feasible power demands at the load node. The red square indicates the maximal power that the load can drain from the power grid. There are multiple operating points if $0<P_c<0.75$. The Jacobian $\pdd {V_L}{P_c}$ is not defined in the red point. The corresponding operating point is long-term voltage semi-stable, but not long-term voltage stable. The solid curve depicts a one-to-one correspondence between the feasible power demands and the long-term voltage semi-stable operating points.\vspace{-10pt}}
\end{figure}

Eq. \eqref{eqn:one node example:power demand inequality} of Example~\ref{example:one node example} shows that \eqref{eqn:dc power flow equation} is not always feasible for each $P_c$ for $n=1$. We will show that the same is true for $n>1$ by studying the maximal total amount of power that can be transported to the load nodes.
\begin{definition}
For a feasible power demand $P_c\in \mc F$, the \emph{total feasible power demand} is the sum $\mb 1\T P_c$ of the power demands at the loads.
\end{definition}
\begin{definition}\label{definition:maximiaing feasible power demand}
A \emph{maximizing feasible power demand} is a feasible power demand $P_{\text{max}}\in \mc F$ that maximizes the total feasible power demand. Thus for all $P_c\in \mc F$ it satisfies
\begin{align}\label{eqn:power flow necessary condition}
\mb 1\T P_c \le \mb 1\T P_{\text{max}}.
\end{align}
\end{definition}
\begin{lemma}\label{lemma:maximizing power demand}
There is a unique maximizing feasible power demand $P_{\text{max}}\in \mc F$. It is given by
\begin{align}\label{eqn:max total power demand}
P_{\text{max}} = \tfrac 1 4 [V_L^*]\mc I_L^* \gneqq \mb 0.
\end{align}
The unique operating point corresponding to $P_{\text{max}}$ is $\tfrac 1 2 V_L^*$.
\end{lemma}

\begin{proof}
Let $P_c$ be feasible, and let $V_L$ be an associated operating point. Recall from \eqref{eqn:total power demand} that the total feasible power demand $\mb 1\T P_c$ satisfies
\begin{align*}
\mb 1\T P_c &= - V_L\T Y_{LL}(V_L - V_L^*).
\end{align*}
By completing the squares we find that
\begin{align*}
\mb 1\T P_c &= - (V_L^* - \tfrac 1 2 V_L)\T Y_{LL}(V_L^* - \tfrac 1 2 V_L) + \tfrac 1 4 {V_L^*}\T Y_{LL} V_L^*.
\end{align*}
Since $Y_{LL}$ is positive definite, it follows that 
\begin{align}\label{eqn:power flow necessary condition 1}
\mb 1 \T P_c \le \tfrac 1 4 {V_L^*}\T Y_{LL} V_L^*,%
\end{align} 
with equality if and only if $V_L = \tfrac 1 2 V_L^*$.
This implies that equality in \eqref{eqn:power flow necessary condition 1} holds if and only if
\begin{align*}
P_c = -P_L(\tfrac 1 2 V_L^*) = \tfrac 1 4 [V_L^*]Y_{LL}V_L^* = \tfrac 1 4 [V_L^*]\mc I_L^*,
\end{align*}
where we have substituted \eqref{eqn:open-circuit voltages}.
The above implies that there is a unique $P_{\text{max}}$ given by \eqref{eqn:max total power demand}, and corresponds to the unique operating point $\tfrac 1 2 V_L^*$.
Lemma~\ref{lemma:positive open-circuit voltages} implies that $P_{\text{max}}\gneqq \mb 0$, since $V_L^*>\mb 0$  and $\mc I_L^*\gneqq \mb 0$. %
\end{proof}

We remark that if a load node $i$ does not share a line with a source node, then $(\mc I_L^*)_i = (-Y_{LS}V_S)_i=0$, and $P_{\text{max},i}=0$.

The inequality \eqref{eqn:power flow necessary condition} describes a closed half-space in the space of power demands, and is a necessary condition for the feasibility of \eqref{eqn:dc power flow equation}. This condition coincides with the inclusion
\begin{align}\label{eqn:power flow necessary condition set}
\mc F \subset \set{P_c}{\mb 1 \T P_c \le \mb 1 \T P_{\text{max}} }.
\end{align}
We observe that \eqref{eqn:power flow necessary condition set} generalizes \eqref{eqn:one node example:F} for $n\ge 1$.
Since there is a unique maximizing feasible power demand by Lemma~\ref{lemma:maximizing power demand}, equality in \eqref{eqn:power flow necessary condition} only holds for $P_{\text{max}}$, and the inclusion in \eqref{eqn:power flow necessary condition set} strict for $n>1$.

The converse of \eqref{eqn:power flow necessary condition} states that, if $P_c$ is such that $\mb 1 \T P_c > \mb 1 \T P_{\text{max}}$, then no solution to \eqref{eqn:dc power flow equation} exists. The existence of $P_{\text{max}}$ therefore once more shows that the DC power flow equations with constant-power loads are not always feasible.
The following example illustrates \eqref{eqn:power flow necessary condition set}.

\begin{figure}[b]
\centering
\includegraphics[scale=.85]{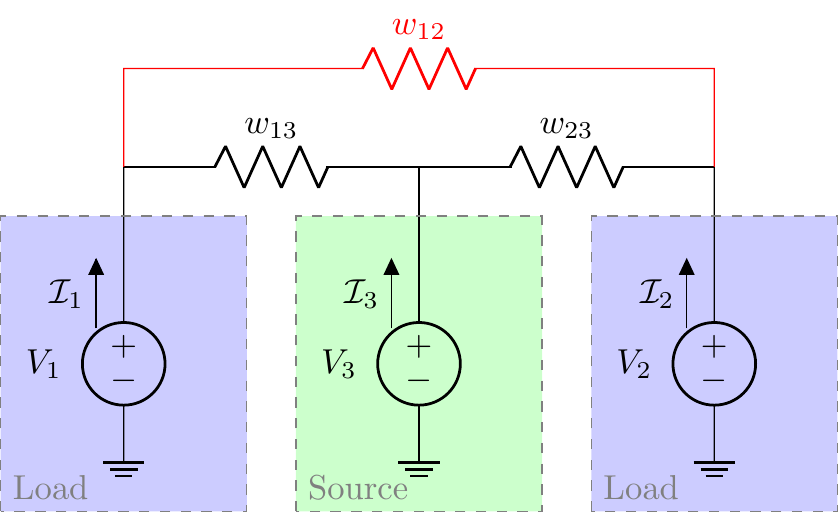}
\caption{\label{figure:two load nodes schematic}
A schematic depiction of a power grid with two load nodes ($n=2$, $m=1$).}
\end{figure}
\begin{figure}[b]
\centering
\includegraphics[scale=.55]{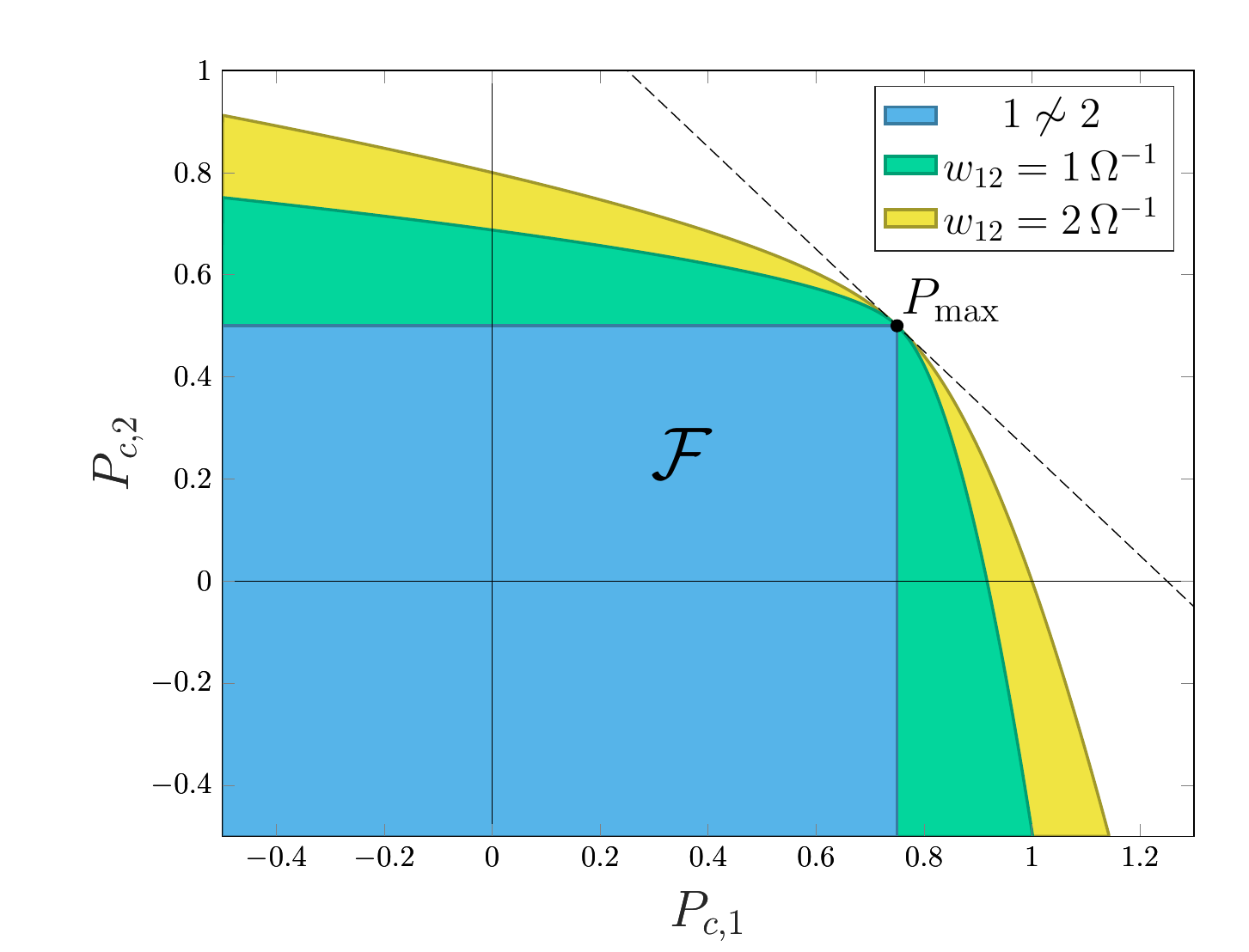}
\caption{\label{figure:two load nodes feasible power injections}
Plots of $\mc F$, the set of feasible power demands, for the power grid in Figure~\ref{figure:two load nodes schematic} with different values of $w_{12}$, where $w_{13} = 3~\Omega\inv$, $w_{23} = 2~\Omega\inv$ and $V_S = 1 \on{V}$. The dashed line is the set of the points for which equality in \eqref{eqn:power flow necessary condition} holds.}
\end{figure}
\begin{example}[two loads, one source case]\label{example:two node example}
Consider the DC power grid with two loads ($n=2$) and one source ($m=1$) depicted in Figure~\ref{figure:two load nodes schematic}. Figure~\ref{figure:two load nodes feasible power injections} gives the feasible power demands when we let $w_{13}=3\on\Omega\inv$, $w_{23}=2\on\Omega\inv$ and $V_S=1 \on V$, and vary the conductance $w_{12}$. It can be shown that $V_L^* = \mb 1 \on V$.\\
First, we disregard the red line between node 1 and node 2 (\ie, $1\not\sim2$), or equivalently take $w_{12} = 0\on\Omega\inv$.
The absence of the red line implies that $Y_{LL}$ is (block) diagonal, and Problem~\ref{problem:dc power flow problem} reduces to two copies of Example~\ref{example:one node example}.
From \eqref{eqn:one node example:power demand inequality} it follows that $P_c$ is feasible if and only if $P_{c,1}\le 0.75$ and $P_{c,2} \le 0.5$, which corresponds to the blue rectangle in Figure~\ref{figure:two load nodes feasible power injections}. 
Next, we consider the red line between loads 1 and 2. We observe from the same figure that increasing $w_{12}$ will result in a larger set of feasible power demands, as indicated by the green and yellow areas. 
The dashed line are the points for which equality in \eqref{eqn:power flow necessary condition} holds.
We note that these sets lie below the dashed line, and intersect the line only at the point $P_{\text{max}}$, which illustrates \eqref{eqn:power flow necessary condition set}. \hfill\QED
\end{example}

Figure~\ref{figure:polyhedral sufficient condition} relates the sufficient conditions of \cite{simpson2016voltage} and \cite{bolognani2015existence} to the feasible power demands of the DC power grid depicted in Figure~\ref{figure:two load nodes schematic}. Figures~\ref{figure:two load nodes feasible power injections} and~\ref{figure:polyhedral sufficient condition} suggest some properties of the set of feasible power demands $\mc F$, which we will prove in this paper:
\begin{itemize}
\item The set $\mc F$ is convex (\mainresult{convexity of F}) (See also \cite{dymarsky2014convexity}).
\item Each hyperplane which is tangent to the boundary of $\mc F$ (such as the dashed line in Figure~\ref{figure:two load nodes feasible power injections}) gives a necessary condition for feasibility (\mainresult{convexity of F}) (See also \cite{barabanov2016}).
\item If $y\in \mc F$ and $\hat y\le y$, then also $\hat y\in\mc F$ (\mainresult{power demand domination}).
\item The convex hull of the points on the boundary of $\mc F$ lead to a sufficient condition for a power demand to be feasible. In particular, the convex hull of $\mb 0$, $P_{\text{max}}$ and the points where the axes intersect the boundary of $\mc F$ forms a polyhedral subset of $\mc F$. The interior of this set describes the sufficient condition from \cite{simpson2016voltage} (\mainresult{sufficient conditions}). See also Figure~\ref{figure:polyhedral sufficient condition}. %
\item The ball with the smallest radius such that it touches the boundary of the condition in \cite{simpson2016voltage} is contained in $\mc F$. This subset describes the sufficient condition of \cite{barabanov2016} (\mainresult{sufficient conditions}); see again Figure~\ref{figure:polyhedral sufficient condition}.
\end{itemize}
\begin{figure}[t]
\centering
\includegraphics[scale=.6]{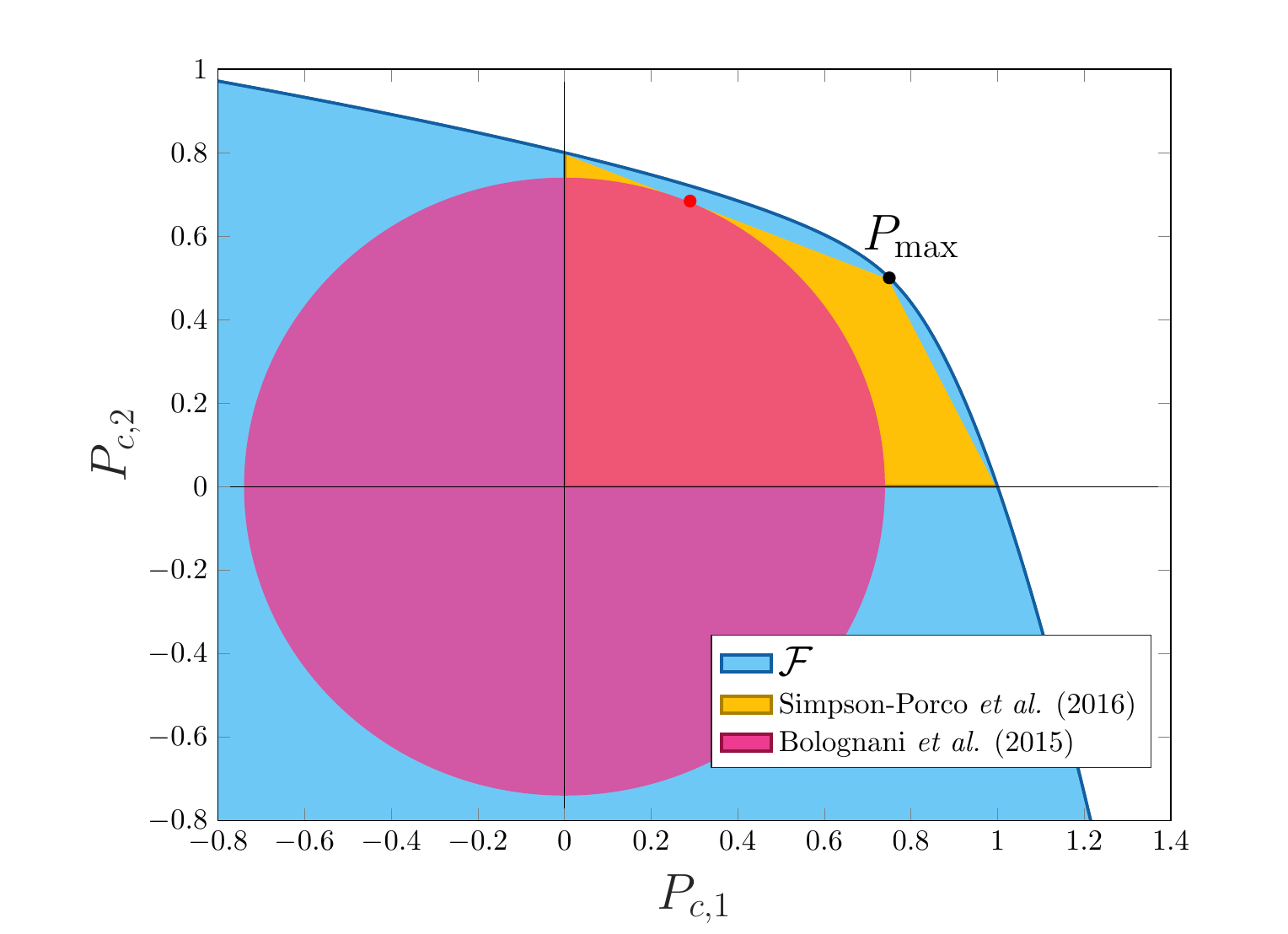}
\caption{\label{figure:polyhedral sufficient condition}
A plot of the set $\mc F$ for the power grid in Figure~\ref{figure:two load nodes schematic}, where $w_{12} = 2~\Omega\inv$, $w_{13} = 3~\Omega\inv$, $w_{23} = 2~\Omega\inv$ and $V_S = 1 \on{V}$. The yellow area is the %
set described by the sufficient condition in \cite{simpson2016voltage}. The red area is the sufficient condition from \cite{bolognani2015existence}. The boundaries of the two conditions intersect in the red point.\vspace{-15pt}}
\end{figure}
We also observe in Figure~\ref{figure:two load nodes feasible power injections} that $P_{\text{max}}$ does not change when $w_{12}$ is changed, and increasing $w_{12}$ leads to nested\footnote{Nested with respect to inclusion.} sets $\mc F$ of feasible power demands. 
We remark that this is %
not true in general. The analysis of this phenomenon is beyond the scope of this paper.
\section{A geometric framework for DC power flow feasibility with constant-power loads}\label{section:problem analysis}
In this section we establish a geometric framework for the feasibility of DC power flow with constant-power loads.
This section is structured as follows.
In Section~\ref{subsection:jacobian of f} we show that every operating point is uniquely associated to a Z-matrix: the Jacobian of $P_c$ at that operating point. Moreover, we show that an operating point is long-term voltage stable if and only if the Jacobian of $P_c$ is a nonsingular M-matrix.
Section~\ref{subsection:parametrization of D} uses this characterization to obtain a parametrization of the set of long-term voltage stable operating points $\mc D$. In particular, we parametrize the boundary of $\mc D$ by a set $\Lambda_1$. %
In Section~\ref{subsection:convex hull of F} we study the convex hull of $\mc F$ and show that $\Lambda_1$ also parametrizes the boundary of $\conv{\mc F}$. This establishes a one-to-one correspondence between the boundary of $\mc D$ and the boundary of $\conv{\mc F}$.
Our main results are stated in Section~\ref{subsection:convexity of F}, in which we prove that for each feasible power demand there exists a unique long-term voltage semi-stable operating point. In addition we present an explicit method for computing this operating point and show that the set of feasible power demands is closed and convex.  %
Finally, in Section~\ref{subsection:necessary and sufficient condition} we prove that the LMI condition in \cite{barabanov2016} is necessary and sufficient for the feasibility of a vector of power demands, and present a similar necessary and sufficient LMI for the feasibility of a vector of power demands under small perturbation.

\subsection{Relating operating points to the Jacobian of $P_c$}\label{subsection:jacobian of f}
Recall that Z-matrices, M-matrices and irreducible matrices were defined in Definitions~\ref{definition:Z-matrix}-\ref{definition:irreducible matrix}.

\begin{proposition}\label{proposition:Z-matrix perron result}
Let $A$ be an irreducible Z-matrix. There is a unique eigenvalue $r$ of $A$ with smallest (\ie, ``most negative'') real part. The eigenvalue $r$, known as the \emph{Perron root}, is real and simple. A corresponding eigenvector $v$, known as a \emph{Perron vector}, is unique up to scaling, and can be chosen such that $v>\mb 0$.
\end{proposition}

We include a short proof, as we were unable to find a reference to this result in this exact formulation.

\begin{proof}
Let $B$ be a nonnegative matrix and $s$ a scalar such that $A=sI-B$.
Irreducibility is independent of the diagonal elements of a matrix.
Hence, since $A$ is irreducible, so is $B$. 
Let $\rho(B)$ denote the spectral norm of $B$.
By the Perron-Frobenius Theorem \cite[Thm. 4.8]{fiedler1986special}, $\rho(B)$ is a simple eigenvalue of $B$ and there exists a positive eigenvector $v$ so that $Bv = \rho(B)v$.
Hence, $r=s-\rho(B)$ is a simple eigenvalue of $A$. The corresponding eigenvector $v$ is unique up to scaling. 
\end{proof}
Appendix~\ref{subsection:matrix theory} lists a number of useful results concerning Z-matrices, M-matrices and irreducible matrices.

The Jacobian of $P_c$ at $\tilde V_L$ is given by
\begin{align}\label{eqn:jacobian of f}
\pdd {P_c}{V_L}(\tilde V_L) = [Y_{LL}(V_L^* - \tilde V_L)] - [\tilde V_L]Y_{LL}.
\end{align}
Recall that a vector $\tilde V_L$ qualifies as a vector of voltage potentials only if $\tilde V_L>\mb 0$. 
The following lemma shows that the matrix \eqref{eqn:jacobian of f} has a particular structure if (and only if) $\tilde V_L>\mb 0$, and that each such matrix is unique for $\tilde V_L>\mb 0$.
\begin{lemma}\label{lemma:jacobian Z-matrix}
The matrix $-\pdd {P_c}{V_L}(\tilde V_L)$ (\ie, the Jacobian of $-P_c$ at $\tilde V_L$) is an irreducible Z-matrix if and only if $\tilde V_L>\mb 0$. %
The map $\tilde V_L \mapsto \pdd {P_c} {V_L} (\tilde V_L)$ is injective for $\tilde V_L > \mb 0$.
\end{lemma}
\begin{proof}
($\Rightarrow$): 
Let $i\in\boldsymbol n$ and note that we have 
\begin{multline}\label{eqn:jacobian Z-matrix:submatrix of jacobian}
\l(-\tpdd {P_c} {V_L}(x)\r)_{[i,i\comp]} = ([x]Y_{LL} - [Y_{LL}(V_L^* - x)])_{[i,i\comp]} \\ =([x]Y_{LL})_{[i,i\comp]}  = x_i (Y_{LL})_{[i,i\comp]}. 
\end{multline}
If $x_i= 0$, then $\l(-\tpdd {P_c} {V_L}(x)\r)_{[i,i\comp]}=0$ by \eqref{eqn:jacobian Z-matrix:submatrix of jacobian}, which violates the irreducibility of the $-\pdd {P_c} {V_L}(x)$. Hence $x_i\neq 0$. 
Since $Y_{LL}$ is a Z-matrix, we have $(Y_{LL})_{[i,i\comp]}\le 0$. Since $Y_{LL}$ is irreducible we know that $(Y_{LL})_{[i,i\comp]}\neq 0$, by definition. It follows that there exists at least one negative element in $(Y_{LL})_{[i,i\comp]}$. Hence, there exists a $j\neq i$ so that $(Y_{LL})_{ij}<0$. It follows from \eqref{eqn:jacobian Z-matrix:submatrix of jacobian} that 
\begin{align}\label{eqn:jacobian Z-matrix:nonnegative element}
\l(-\tpdd {P_c} {V_L}(x)\r)_{ij} = x_i (Y_{LL})_{ij}.
\end{align}
The right-hand side of \eqref{eqn:jacobian Z-matrix:nonnegative element} is nonzero, $(Y_{LL})_{ij}$ is negative, and the left-hand side of \eqref{eqn:jacobian Z-matrix:nonnegative element} is nonpositive since $-\tpdd {P_c} {V_L}(x)$ is a Z-matrix. This implies that $x_i$ is positive. Hence $x>\mb 0$.

($\Leftarrow$):
The matrix $Y_{LL}$ is an irreducible Z-matrix. 
Since $x>\mb 0$, also $[x]Y_{LL}$ is an irreducible Z-matrix, by \ref{proposition:diagonal properties:irreducible multiplication} and \ref{proposition:diagonal properties:Z-matrix multiplication} of Proposition~\ref{proposition:diagonal properties}.
Consequently, $[x]Y_{LL} - [Y_{LL}(V_L^* - x)]$ is an irreducible Z-matrix, by \ref{proposition:diagonal properties:irreducible addition} and \ref{proposition:diagonal properties:Z-matrix addition} of Proposition~\ref{proposition:diagonal properties}. 

Let $x,z>\mb 0$ satisfy $\pdd {P_c} {V_L}(x) = \pdd {P_c} {V_L}(z)$. By \eqref{eqn:jacobian Z-matrix:nonnegative element} this implies that for all $i$ there exists a $j$ such that $x_i (Y_{LL})_{ij} = z_i (Y_{LL})_{ij}$. Since $(Y_{LL})_{ij}\neq 0$, it follows that $x_i=z_i$, and hence $x=z$.
\end{proof}

Lemma~\ref{lemma:jacobian Z-matrix} states that the Jacobian of $P_c$ at an operating point $\tilde V_L > \mb 0$ is unique to $\tilde V_L$. 
This implies that operating points are uniquely identified by properties of the associated Jacobian.
The following result identifies all long-term voltage stable operating points (see Definition~\ref{definition:long-term voltage stable}) by means of properties of the associated Jacobian.
\begin{proposition}\label{proposition:jacobian M-matrix}
The set $\mc D$ of long-term voltage stable operating points equals
\begin{align*}
\mc D = \set{\tilde V_L >\mb 0}{-\pdd {P_c} {V_L} (\tilde V_L) \text{ is a nonsingular M-matrix}}.
\end{align*}
\end{proposition}
\begin{proof}
Let $\tilde V_L>\mb 0$ be an operating point associated to some vector of power demands $\tilde P_c$.
Recall from Definition~\ref{definition:long-term voltage stable} that $\tilde V_L$ is long-term voltage stable if $\pdd {P_c} {V_L}(\tilde V_L)$ is nonsingular and $\pdd {P_c} {V_L}(\tilde V_L)\inv = \pdd {V_L}{P_c} (\tilde P_c)$ is a matrix with negative elements.
Since $\tilde V_L>\mb 0$, $-\pdd {P_c} {V_L}(\tilde V_L)$ is a Z-matrix by Lemma~\ref{lemma:jacobian Z-matrix}. If follows by \cite[Thm. 5.12]{fiedler1986special} that  $-\pdd {P_c} {V_L}(\tilde V_L)\inv > 0$ if and only if $-\pdd {P_c} {V_L}(\tilde V_L)$ is a nonsingular M-matrix. 
\end{proof}

Recall from Definition~\ref{definition:long-term voltage semi-stable} that an operating point is long-term voltage semi-stable if it lies in the the closure of $\mc D$.
Proposition~\ref{proposition:jacobian M-matrix} implies the following characterization of such operating points.
\begin{corollary}\label{corollary:boundary of D}
The closure and boundary of $\mc D$ satisfies
\begin{align*}
\cl{\mc D} &= \set{\tilde V_L>\mb 0}{-\pdd {P_c} {V_L} (\tilde V_L) \text{ is an M-matrix}};\\
\partial \mc D &= \set{\tilde V_L>\mb 0}{-\pdd {P_c} {V_L} (\tilde V_L) \text{ is a singular M-matrix}}.%
\end{align*}
\end{corollary}
The proof follows directly from Proposition~\ref{proposition:perron root M-matrix}.

\begin{remark}\label{remark:equivalent definitions of long-term voltage stability}
Many equivalent characterizations of long-term voltage stable operating points may be derived from Proposition~\ref{proposition:jacobian M-matrix}.
Indeed, the paper \cite{PLEMMONS1977175} lists numerous equivalent conditions for when a Z-matrix is a nonsingular M-matrix.
In particular, it follows from property $\mathrm J_{29}$ of \cite{PLEMMONS1977175} together with Proposition~\ref{proposition:diagonal properties}.\ref{proposition:diagonal properties:M-matrix multiplication} that if $\tilde V_L>\mb 0$, and hence $-\pdd {P_c} {V_L} (\tilde V_L)$ is a Z-matrix by Lemma~\ref{lemma:jacobian Z-matrix}, then $-\pdd {P_c} {V_L} (\tilde V_L)$ is a nonsingular M-matrix if and only if $[\tilde V_L]\inv \pdd {P_c} {V_L} (\tilde V_L)$ is Hurwitz stable. This shows that $\tilde V_L>\mb 0$ is long-term voltage stable if and only if $[\tilde V_L]\inv \pdd {P_c} {V_L} (\tilde V_L)$ is Hurwitz stable. The latter property coincides with the definition of voltage-regularity found in \cite{matveev2020tool}. Alternatively, without invoking Proposition~\ref{proposition:diagonal properties}.\ref{proposition:diagonal properties:M-matrix multiplication} it follows that $\tilde V_L>\mb 0$ is long-term voltage stable if and only if $\pdd {P_c} {V_L} (\tilde V_L)$ is Hurwitz stable.
Similarly it can be shown that $\tilde V_L>\mb 0$ is long-term voltage semi-stable if and only if $\pdd {P_c} {V_L} (\tilde V_L)$ is Hurwitz semi-stable\footnote{By a Hurwitz semi-stable matrix we mean a matrix for which all its eigenvalues have negative real part, with the possible exception of a semi-simple eigenvalue 0. 
The eigenvalue 0 of a singular symmetric M-matrices is semisimple.
}.
\end{remark}

\subsection{A parametrization of $\mc D$}\label{subsection:parametrization of D}
Proposition~\ref{proposition:jacobian M-matrix} and Corollary~\ref{corollary:boundary of D} allow us to deduce a parametrization for the set $\mc D$ of all long-term voltage stable operating points. 
Such a parametrization gives a constructive method to determine where such operating points lie in the voltage domain, as opposed to testing at which operating points of interest the Jacobian of $P_c$ is Hurwitz stable (see Remark~\ref{remark:equivalent definitions of long-term voltage stability}).

We introduce the following definitions.
For a vector $\lambda\in\RR^n$ we introduce the $n\times n$ matrix 
\begin{align*}
h(\lambda) := \tfrac 1 2 ([\lambda]Y_{LL} + Y_{LL}[\lambda]).
\end{align*}
Note that $h(\mb 1) = Y_{LL}$, and that 
\begin{align}\label{eqn:rewriting quadratic}
x\T h(\lambda) x = x\T [\lambda] Y_{LL} x = \lambda\T [x]Y_{LL}x.
\end{align}
In addition we define the set
\begin{align}\label{eqn:Lambda}
\Lambda := \set{\lambda}{%
h(\lambda) \text{ is positive definite}}.
\end{align}
The set $\Lambda$ is studied in Appendix~\ref{subsection:closure of Lambda}. In particular, Lemma~\ref{lemma:Lambda convex} shows that $\Lambda$ is convex, and Lemma~\ref{lemma:Lambda positive} shows that $\Lambda$ lies in the positive orthant. 

The following theorem extends Lemma~\ref{lemma:jacobian Z-matrix}, and allows us to parametrize the sets $\mc D$, $\cl{\mc D}$, and $\partial\mc D$.
\begin{lemma}\label{lemma:characterization pddf M-matrix}
Let $r\in\RR$ and $\lambda\in\RR^n$ such that $r\ge 0$ and $\lambda>\mb 0$.
The Jacobian $-\pdd {P_c} {V_L}(\tilde V_L)\T$ is an irreducible M-matrix with Perron root $r$ and Perron vector $\lambda$ if and only if $h(\lambda)$ is positive definite (\ie, $\lambda\in\Lambda$) and $\tilde V_L$ satisfies 
\begin{align}\label{eqn:characterization pddf M-matrix:definition of x}
\tilde V_L = \tfrac 1 2 h(\lambda)\inv[\lambda](\mc I_L^* + r \mb 1),
\end{align}
in which case we have $\tilde V_L>\mb 0$.
\end{lemma}
\begin{proof}
($\Rightarrow$):
The matrix $Y_{LL}$ is an irreducible Z-matrix and $\lambda>\mb 0$, and so $h(\lambda)$ is an irreducible Z-matrix by Propositions~\ref{proposition:diagonal properties} and \ref{proposition:sum of irreducible Z-matrices}.
We let $s$ and $v>\mb 0$ denote respectively the Perron root and a Perron vector of $h(\lambda)$.
The matrix $-\pdd {P_c} {V_L} (\tilde V_L)\T$ is an M-matrix, and hence a Z-matrix. By Lemma~\ref{lemma:jacobian Z-matrix} we have $\tilde V_L>\mb 0$.
Using the fact that $(r,\lambda)$ is an eigenpair to $-\pdd {P_c} {V_L}(\tilde V_L)\T$, we observe that %
\begin{align*}
r\lambda &= -\pdd {P_c} {V_L} (\tilde V_L)\T \lambda = -([Y_{LL}(V_L^*-\tilde V_L)] - Y_{LL}[\tilde V_L])\lambda\\
&=-[\lambda]Y_{LL}V_L^* + [\lambda]Y_{LL} \tilde V_L + Y_{LL}[\lambda]\tilde V_L \\&= -[\lambda]\mc I_L^* + 2 h(\lambda)\tilde V_L.
\end{align*}
By rearranging the terms, it follows that
\begin{align}\label{eqn:characterization pddf M-matrix:x relation}
[\lambda]\mc I_L^* + r\lambda = 2h(\lambda)\tilde V_L.
\end{align}
Multiplying \eqref{eqn:characterization pddf M-matrix:x relation} by $v\T$ results in 
\begin{align}\label{eqn:characterization pddf M-matrix:main argument}
v\T ([\lambda]\mc I_L^* + r\lambda) = 2v\T h(\lambda)\tilde V_L = 2s v\T \tilde V_L.
\end{align}
Since $\tilde V_L>\mb 0$, $v>\mb 0$, $\lambda>\mb 0$, $r\ge 0$ and $\mc I_L^*\gneqq\mb 0$, it follows\linebreak that the left hand side of \eqref{eqn:characterization pddf M-matrix:main argument} is positive. Since  $v\T \tilde V_L$ is also positive, it follows from \eqref{eqn:characterization pddf M-matrix:main argument} that the Perron root $s$ is positive. Hence, $h(\lambda)$ is a nonsingular M-matrix by Proposition~\ref{proposition:perron root M-matrix}, and \eqref{eqn:characterization pddf M-matrix:definition of x} follows from \eqref{eqn:characterization pddf M-matrix:x relation}. Since $h(\lambda)$ is a symmetric nonsingular M-matrix, it is positive definite and $\lambda\in\Lambda$. 

($\Leftarrow$): If $\lambda\in\Lambda$, then $\lambda>\mb 0$ by Lemma~\ref{lemma:Lambda positive}. The rest of the proof follows by reversing the steps of the ``$\Rightarrow$''-part.
\end{proof}

Lemma~\ref{lemma:characterization pddf M-matrix} allows for an explicit parametrization of the set $\mc D$ by $\Lambda$, and without relying on properties of the Jacobian of $P_c$.
Note that \eqref{eqn:characterization pddf M-matrix:definition of x} is invariant under scaling of $\lambda$, and hence the vectors $\lambda$ may be normalized. 
For this purpose we define
\begin{align}\label{eqn:Lambda_1}
\Lambda_1:=\Lambda\cap\set{\lambda>\mb 0}{\|\lambda\|_1=\mb 1\T \lambda = 1},
\end{align}
which is a convex set, as it is the intersection of convex sets.
Appendix~\ref{subsection:closure of Lambda} lists several properties of the closure of $\Lambda_1$.

\begin{theorem}[\mainresult{parametrization of D}]\label{theorem:parametrization of D}
The set $\mc D$ of all long-term voltage stable operating point, its closure $\cl{\mc D}$ and its boundary $\partial\mc D$ are parametrized by
\begin{align*}
\mc D &= \set{\tfrac 1 2 h(\lambda)\inv[\lambda](\mc I_L^* + r \mb 1)}{\lambda\in\Lambda_1, r>0}\\
\cl{\mc D} &= \set{\tfrac 1 2 h(\lambda)\inv[\lambda](\mc I_L^* + r \mb 1)}{\lambda\in\Lambda_1, r\ge 0}\\
\partial\mc D &= \set{\tfrac 1 2 h(\lambda)\inv[\lambda]\mc I_L^*}{\lambda\in\Lambda_1}.
\end{align*}
Furthermore, the map
\begin{align}\label{theorem:parametrization of D:map}
(\lambda,r) \mapsto \tfrac 1 2 h(\lambda)\inv[\lambda](\mc I_L^* + r \mb 1)
\end{align}
from $\Lambda_1\times\RR_{\ge 0}$ to $\cl{\mc D}$ is a bicontinuous map, and the sets $\mc D$, $\cl{\mc D}$ and $\partial\mc D$ are simply connected.
\end{theorem}
\begin{proof}
Proposition~\ref{proposition:perron root M-matrix} states that a Z-matrix is a (nonsingular/singular) M-matrix if and only if its Perron root $r$ is nonnegative (positive/zero).
Proposition~\ref{proposition:jacobian M-matrix} and Corollary~\ref{corollary:boundary of D} together with Lemma~\ref{lemma:characterization pddf M-matrix} imply that the vector \eqref{eqn:characterization pddf M-matrix:definition of x} with $\lambda\in\Lambda$ lies in $\mc D$, $\cl{\mc D}$ or $\partial\mc D$ if and only if $r$ in \eqref{eqn:characterization pddf M-matrix:definition of x} satisfies respectively $r>0$, $r\ge 0$ or $r=0$.

The map \eqref{theorem:parametrization of D:map} is a continuous bijection from $\Lambda_1\times \RR_{\ge 0}$ to $\cl{\mc D}$, which follows from Lemma~\ref{lemma:characterization pddf M-matrix} and Corollary~\ref{corollary:boundary of D}.
The inverse of the map \eqref{theorem:parametrization of D:map} is described taking $x\in\cl{\mc D}$ and computing the Perron vector $\lambda > \mb 0$ and the Perron root $r$ of $-\pdd {P_c} {V_L}(x)$. By Proposition~\ref{proposition:continuity of perron root and vector}, the Perron root and Perron vector of $-\pdd {P_c} {V_L}(x)$ are continuous in $x$. Hence the inverse of the map \eqref{theorem:parametrization of D:map} is also continuous. %

The set $\Lambda_1\times\RR_{\ge 0}$ is convex, and is therefore simply connected, which is a topological property. Topological properties are preserved by bicontinuous maps, and thus $\cl{\mc D}$ is also simply connected.
The same holds for $\Lambda_1\times\RR_{> 0}$ and $\Lambda_1\times \{0\}$, and hence $\mc D$ and $\partial \mc D$ are simply connected.
\end{proof}

\subsection{The convex hull of $\mc F$ and its boundary}\label{subsection:convex hull of F}
In order to study the set $\mc F$ of feasible power demands, we will be studying its convex hull and its boundary. 
In particular, we show that there is a one-to-one correspondence between points in the boundary of $\mc D$ and points in $\partial\conv{\mc F}$, the boundary of the convex hull of $\mc F$.

To simplify notation, we define for $\lambda\in\Lambda$ the map %
\begin{align}\label{eqn:definition of m}
\varphi (\lambda):=\tfrac 1 2 h(\lambda)\inv [\lambda]\mc I_L^*.
\end{align}
Note that $\varphi (\lambda)$ is invariant under scaling of $\lambda$.

The boundary of $\mc D$ is by definition the set of the operating points which are long-term voltage semi-stable, but not long-term voltage stable.
It follows from Theorem~\ref{theorem:parametrization of D} that $\partial\mc D$ satisfies %
\begin{align}\label{eqn:boundary of D}
\partial \mc D =\set{\varphi (\lambda)}{\lambda\in\Lambda_1} = \varphi (\Lambda_1).
\end{align}

To study the convex hull of $\mc F$ and its boundary, we make use of the following two identities involving $P_c(V_L)$.
\begin{lemma}\label{lemma:sum formula f}
For $x,z\in \RR^n$ we have
\begin{align}\label{eqn:sum formula f}
P_c(x+z) = P_c(x) + \pdd {P_c} {V_L} (x)z - [z]Y_{LL} z.
\end{align}
\end{lemma}
\begin{proof*}
We write out the formula using \eqref{eqn:definition of f} and use \eqref{eqn:jacobian of f}:
\begin{align*}
\hspace{.5em}&P_c(x+z) = [x+z]Y_{LL}(V_L^* - x - z) = [x]Y_{LL}(V_L^* - x)\\
&\quad\quad + [z]Y_{LL}(V_L^* - x) - [x]Y_{LL}z - [z]Y_{LL}z\\
&\quad= P_c(x) + [Y_{LL}(V_L^* - x)]z - [x]Y_{LL}z - [z]Y_{LL}z \\
&\quad= P_c(x) + \pdd {P_c} {V_L}(x)z - [z]Y_{LL}z. \hspace{8.8em}\QED
\end{align*}
\end{proof*}%

The matrix $h(\lambda)$ for $\lambda\in\Lambda$ is positive definite by definition, and therefore induces the vector norm
\begin{align}\label{eqn:definition norm}
\|x\|_{h(\lambda)} := \sqrt{x\T h(\lambda) x}.
\end{align}
This vector norm is related to $\lambda\T P_c(x)$ for $\lambda\in\Lambda$
by the following lemma.
\begin{lemma}\label{lemma:lambda f}
Let $\lambda\in\Lambda$. For each $x\in\RR^n$ we have
\begin{align}
\lambda\T P_c (x) = \|\varphi (\lambda)\|_{h(\lambda)}^2 - \|\varphi (\lambda) - x\|_{h(\lambda)}^2.\label{eqn:lambda f}
\end{align}
Moreover, we have 
\begin{align}\label{eqn:lambda f inequality}
\lambda\T P_c(x) \le \|\varphi (\lambda)\|_{h(\lambda)}^2 = \lambda\T P_c(\varphi (\lambda)) ,
\end{align}
with equality if and only if $x=\varphi (\lambda)$. Consequently, we have $P_c(x) = P_c(\varphi (\lambda))$ if and only if $x=\varphi (\lambda)$.
\end{lemma}
\begin{proof}
Using \eqref{eqn:definition of f}, \eqref{eqn:open-circuit voltages}, \eqref{eqn:definition of m} and \eqref{eqn:rewriting quadratic}, we verify that
\begin{align*}
\lambda\T P_c (x) &= \lambda\T [x]Y_{LL}V_L^* - \lambda\T [x] Y_{LL} x\nonumber\\
&= x\T [\lambda] \mc I_L^* - x\T [\lambda]Y_{LL} x\nonumber\\
&= 2 x\T h(\lambda)\varphi (\lambda) - x\T h(\lambda) x\nonumber\\
&= \varphi (\lambda)\T h(\lambda) \varphi (\lambda) - \varphi (\lambda)\T h(\lambda) \varphi (\lambda) \\&\quad+ 2 x\T h(\lambda)\varphi (\lambda) - x\T h(\lambda) x\nonumber\\
&= \|\varphi (\lambda)\|_{h(\lambda)}^2 - \|\varphi (\lambda) - x\|_{h(\lambda)}^2.%
\end{align*}
Eq. \eqref{eqn:lambda f inequality} follows from \eqref{eqn:lambda f} since $\|\varphi (\lambda) - x\|_{h(\lambda)}\ge 0$, with equality if and only if $x=\varphi (\lambda)$. 
Thus, equality in \eqref{eqn:lambda f inequality} holds if and only if $x=\varphi (\lambda)$.
Finally, if $P_c(x)=P_c(\varphi (\lambda))$, then $\lambda\T P_c(x)=\lambda\T P_c(\varphi (\lambda))$ therefore $x=\varphi (\lambda)$.
\end{proof}%

For a vector $\nu$ such that $\|\nu\|_1=1$ and a scalar $s$ we define the closed half-space
\begin{align*}
H(\nu,s):=\set{y}{\nu\T y \le s},
\end{align*}
which has as boundary the hyperplane 
\begin{align*}
\partial H(\nu,s) =\set{y}{\nu\T y = s}.
\end{align*}
The vector $\nu$ is normal to the boundary of the half-space and points outwards.
\begin{definition}%
A half-space $H(\nu,s)$ is said to \emph{support} a set $S$ if $S \subset H(\nu,s)$ and $\cl S\cap \partial H(\nu,s)$ is nonempty. \Ie, for a given $\nu$, $s$ is the smallest number so that $S \subset H(\nu,s)$. A point in $\cl S\cap \partial H(\nu,s)$ is a \emph{point of support}.
\end{definition}

If the half-space $H(\nu,s)$ supports $\mc F$ and $\tilde P_c\in \mc F$ is a point of support, then $\nu\T \tilde P_c$ maximizes $\nu\T P_c$ for all $P_c\in \mc F$.
For example, let $\nu=\frac 1 n \mb 1$, then the corresponding supporting half-space is given by \eqref{eqn:power flow necessary condition set}. The vector $P_{\text{max}}$ is the unique point of support, as was shown in Lemma~\ref{lemma:maximizing power demand}. See also Figure~\ref{figure:two load nodes feasible power injections}, in which the dashed line corresponds to the boundary of this half-space.

In order to obtain a geometric description of the convex hull of $\mc F$ we aim to apply the following proposition.
\begin{proposition}[Cor. 11.5.1 of \cite{rockafellar1970convex}]\label{proposition:Rockafellar convex hull} %
Let the set $S$ be a subset of $\RR^n$, then
\begin{align*}
\cl{\conv{S}} = \bigcap_{\text{support of }S} H(\nu,s)
\end{align*}
where the intersection is taken over all half-spaces $H(\nu,s)$ which support $S$.
\end{proposition}
To apply Proposition~\ref{proposition:Rockafellar convex hull} to $\mc F$, we identify all supporting half-spaces of the set $\mc F$. We will simultaneously identify the supporting half-spaces of the image of $P_c$, which is %
given by
\begin{align}\label{eqn:definition of image of f}
\im P_c := \set{y}{P_c(x)=y,x\in\RR^n}
\end{align}
and satisfies the inclusion $\mc F\subset \im P_c$.

\begin{theorem}\label{theorem:supporting half-spaces of F} %
Let $\lambda$ be a vector such that $\|\lambda\|_1=1$ and $s$ be a scalar. 
A half-space $H(\lambda,s)$ %
supports $\mc F$ if and only if $\lambda\in \Lambda_1$ and $s=\|\varphi (\lambda)\|_{h(\lambda)}^2$. %
The point $P_c(\varphi (\lambda))\in\mc F$ is the unique point of support.
Moreover, the supporting half-spaces of the set $\mc F$ and the image of $P_c$ coincide.
\end{theorem}

The proof of Theorem~\ref{theorem:supporting half-spaces of F} can be found in Appendix~\ref{appendix:proof of convex hull of F}.
To simplify notation, we define for $\lambda\in\Lambda$ the half-spaces
\begin{align}\label{eqn:definition of H lambda}
H_\lambda &:= H(\lambda,\|\varphi (\lambda)\|_{h(\lambda)}^2).%
\end{align}
Theorem~\ref{theorem:supporting half-spaces of F} states that $H_\lambda$ for $\lambda\in\Lambda$ are all supporting half-spaces of $\mc F$.
Proposition~\ref{proposition:Rockafellar convex hull} therefore allows us to give a direct formula for the closure of the convex hull of $\mc F$.

\begin{corollary}\label{corollary:intersection formula for F}
The closure of the convex hull of $\mc F$ is the intersection of all half-spaces $H_\lambda$ where $\lambda\in\Lambda$, and is equal to the closure of the convex hull of $\im P_c$.
\Ie,
\begin{align}\label{eqn:intersection formula for F}
\cl{\conv{\mc F}} = \cl{\conv{\im P_c}} %
= \bigcap_{\lambda\in\Lambda_1}\nolimits H_\lambda.
\end{align}
\end{corollary}

Now that we have identified the closure of the convex hull of $\mc F$, we may also identify the boundary of this set.
\begin{theorem}\label{theorem:boundary of F}
The map $\lambda\mapsto P_c(\varphi (\lambda))$ for $\lambda\in\Lambda_1$ is one-to-one and parametrizes the boundary of $\conv{\mc F}$. Moreover, the set $\conv{\mc F}$ is closed. %
\end{theorem}

The proof of Theorem~\ref{theorem:boundary of F} can be found in Appendix~\ref{appendix:proof of convex hull of F}.

\begin{corollary}\label{corollary:boundary D F one-to-one}
The sets $\Lambda_1$, $\partial\mc D$ and $\partial\conv{\mc F}$ are in one-to-one correspondence.
In particular, $P_c$ is a one-to-one map from $\partial\mc D$ to $\partial\conv{\mc F}$.
\end{corollary}
\begin{proof}
It follows from Theorem~\ref{theorem:parametrization of D} that $\varphi (\lambda)$ is a one-to-one map from $\Lambda_1$ to $\partial\mc D$, and Theorem~\ref{theorem:boundary of F} states that $\lambda \to P_c(\varphi (\lambda))$ is a one-to-one map from $\Lambda_1$ to $\partial\conv{\mc F}$. Hence, $P_c$ is a one-to-one map from $\partial\mc D$ to $\partial\conv{\mc F}$.
\end{proof}

Corollary~\ref{corollary:boundary D F one-to-one} states that a vector of power demands $\tilde P_c$ that lies on the boundary the convex hull of $\mc F$ corresponds uniquely to an operating point $\tilde V_L$ which is long-term voltage semi-stable but not long-term voltage stable.
The pair $(\tilde P_c, \tilde V_L)$ corresponds to a unique $\lambda\in\Lambda_1$, and the corresponding hyperplane $\partial H_{\lambda}$ intersects $\conv{\mc F}$ only in the unique point of support $\tilde P_c$. Hence, $\partial H_{\lambda}$ is the tangent plane at $\tilde P_c$ of the boundary of $\conv{\mc F}$.
This is observed in Figure~\ref{figure:two load nodes feasible power injections} for $\tilde P_c = P_{\text{max}}$, $\tilde V_L = \frac 1 2 V_L^*$ and $\lambda = \tfrac 1 n \mb 1$, and the same holds for all points on the boundary of $\conv{\mc F}$.

\subsection{One-to-one correspondence between $\mc F$ and $\cl{\mc D}$}%
\label{subsection:convexity of F}

In this section we use Corollary~\ref{corollary:boundary D F one-to-one} to prove that $P_c$ is a one-to-one mapping from $\cl{\mc D}$ to $\conv{\mc F}$, and that therefore $\mc F$ is convex.
This means each feasible power demand is uniquely associated to a long-term voltage semi-stable operating point.
This operating point can be found by solving an initial value problem. 
The following lemma is intrumental in proving these results.

\begin{lemma}\label{lemma:convex combination}
Let $\hat V_L\in\mc D$, define $\hat P_c:=P_c(\hat V_L)\in P_c(\mc D)$ and let $\tilde P_c\in \inter{\conv{\mc F}}$. %
There exists a unique path $\gamma: [0,1]\to \mc D$ so that the convex combination of $\hat P_c$ and $\tilde P_c$ is described by
\begin{align}\label{eqn:convex combination:convex combination}
P_c(\gamma(\theta)) = \theta \tilde P_c + (1-\theta) \hat P_c.
\end{align}
for $0\le \theta\le 1$. The path $\gamma$ solves the initial value problem 
\begin{align}\label{eqn:convex combination:differential equation}
\dot \gamma(\theta) = \l(\pdd {P_c} {V_L} (\gamma(\theta))\r)\inv (\tilde P_c-\hat P_c)%
\end{align}
with initial value $\gamma(0)=\hat V_L$. We have $\tilde P_c = P_c(\gamma(1))$.
\end{lemma}

The proof of Lemma~\ref{lemma:convex combination} can be found in Appendix~\ref{appendix:proofs for convexity of F}.

\begin{theorem}[\mainresult{one-to-one correspondence}]\label{theorem:one-to-one correspondence}
There is a one-to-one correspondence between the long-term voltage semi-stable operating points $\cl{\mc D}$ and the feasible power demands. \Ie, for each $\tilde P_c\in\mc F$ there exists a unique $\tilde V_L\in\cl{\mc D}$ which satisfies $\tilde P_c = P_c(\tilde V_L)$, implying that $\mc F = P_c(\cl{\mc D})$. More explicitly, $\tilde V_L$ is obtained by solving the initial value problem
\begin{align}\label{eqn:one-to-one correspondence:initial value problem}
\dot \gamma(\theta) = \l(\pdd {P_c} {V_L} (\gamma(\theta))\r)\inv \tilde P_c %
\end{align}
for $\gamma: [0,1]\to \RR^n$ with initial value $\gamma(0)=V_L^*$, where the solution $\gamma$ exists, is unique and satisfies $\gamma(1)=\tilde V_L$.
\end{theorem}
\begin{proof}
Note that $\tilde P_c\in\mc F\subset \conv{\mc F}$ and that $V_L^*\in \mc D$.
Suppose $\tilde P_c \in \inter{\conv{\mc F}}$. By taking $\hat V_L = V_L^*$ and $\hat P_c = P_c(V_L^*) = \mb 0$ in Lemma~\ref{lemma:convex combination}, there is a unique $\gamma:[0,1]\to \mc D$ which solves \eqref{eqn:one-to-one correspondence:initial value problem} with $\gamma(0)=V_L^*$ and which satisfies $\tilde P_c = P_c(\gamma(1))$.
Hence we take $\tilde V_L := \gamma(1)\in\mc D$. To show that there is a unique $\tilde V_L\in\mc D$ such that $\tilde P_c = P_c(\tilde V_L)$, suppose that we have $\tilde V_L^\prime\in \mc D$ such that $\tilde P_c = P_c(\tilde V_L^\prime)$.
Then by Lemma~\ref{lemma:convex combination} there is a unique $\gamma^\prime:[0,1]\to \mc D$ which solves 
\begin{align}\label{eqn:one-to-one correspondence:initial value problem 2}
\dot \gamma^\prime(\theta) = \l(\pdd {P_c} {V_L} (\gamma^\prime(\theta))\r)\inv (P_c(V_L^*) - \tilde P_c) %
\end{align} with $\gamma^\prime(0)=\tilde V_L^\prime$,
and satisfies $P_c(\gamma^\prime(1)) = P_c(V_L^*) = \mb 0$. Recall from Proposition~\ref{proposition:jacobian M-matrix} that $\gamma^\prime(1)>\mb 0$ since $\gamma^\prime(1)\in \mc D$. Since $\mb 0 = P_c(\gamma^\prime(1)) = [\gamma^\prime(1)]Y_{LL}(V_L^* - \gamma^\prime(1))$ and $Y_{LL}$ is nonsingular, this implies that $\gamma^\prime(1) = V_L^*$.
But now note that $\hat \gamma(\theta) := \gamma^\prime(1-\theta) + V_L^* -\tilde V_L^\prime$ is a solution to \eqref{eqn:one-to-one correspondence:initial value problem}, since $P_c(V_L^*)=\mb 0$ in \eqref{eqn:one-to-one correspondence:initial value problem 2} and since $\hat \gamma(0) = V_L^*$. 
Since \eqref{eqn:one-to-one correspondence:initial value problem} has a unique solution, it follows that $\hat \gamma = \gamma$, and in particular $\tilde V_L = \gamma(1) = \hat \gamma(1) = \tilde V_L^\prime$, which proves that $\tilde V_L$ is unique. 
Alternatively, if $\tilde P_c \in \partial\conv{\mc F}$ then by Corollary~\ref{corollary:boundary D F one-to-one} there exists a unique $\tilde V_L\in\partial \mc D$ such that $\tilde P_c = P_c(\tilde V_L) \in P_c(\cl{\mc D})$.
It follows from Lemma~\ref{lemma:lambda f} that there is no other operating point $\tilde V_L^\prime$ such that $\tilde P_c = P_c(\tilde V_L^\prime)$. 
The operating point $\tilde V_L$ is also obtained by the initial value problem \eqref{eqn:one-to-one correspondence:initial value problem}, which follows from taking the limit $\tilde P_c\to \partial \mc D$ for $\tilde P_c \in \inter{\conv{\mc F}}$.
\end{proof}

The next theorem proves that the set $\mc F$ of feasible power demands is closed and convex, and gives a geometric characterization of $\mc F$ in terms of the closed half-spaces $H_\lambda$.
\begin{theorem}[\mainresult{convexity of F}]\label{theorem:convexity of F}
The set $\mc F$ of feasible power demands is closed and convex.
Moreover, the set $\mc F$ is the intersection of all half-spaces $H_\lambda$ with $\lambda\in\Lambda_1$ (see \eqref{eqn:definition of H lambda}), and $\mc F$ coincides with the image of $P_c$. \Ie,
\begin{align}\label{eqn:convexity of F:equalities}
\mc F = P_c(\cl{\mc D}) = \cl{\conv{\mc F}} = \bigcap_{\lambda\in\Lambda_1} H_\lambda = \im P_c.
\end{align}
\end{theorem}
\begin{proof}
We will first prove convexity.
By definition we have $P_c(\cl{\mc D}) \subset \mc F \subset \conv{\mc F}$. Hence it suffices to show that $\conv{\mc F} \subset P_c(\cl{\mc D})$.
Let $\tilde P_c\in\conv{\mc F}$. If $\tilde P_c\in\partial \conv{\mc F}$, then Corollary~\ref{corollary:boundary D F one-to-one} implies that there exists $\tilde V_L\in\partial \mc D$ such that $\tilde P_c = P_c(\tilde V_L) \in P_c(\cl{\mc D})$.
Alternatively, if $\tilde P_c\in\inter{\conv{\mc F}}$, then by Lemma~\ref{lemma:convex combination} there exists a path $\gamma: [0,1]\to\mc D$ such that $\gamma(0)=V_L^*\in\mc D$ and \eqref{eqn:convex combination:convex combination} holds. In particular, \eqref{eqn:convex combination:convex combination} implies that $\tilde P_c = P_c(\gamma(1))\in P_c(\mc D)$. 
Thus $\conv{\mc F} \subset P_c(\cl{\mc D})$, and thus $\mc F$ is convex.
Corollary~\ref{corollary:intersection formula for F}, Theorem~\ref{theorem:boundary of F} and the convexity of $\mc F$ further imply that
\begin{align}\label{eqn:corollary:implications of convexity}
\bigcap_{\lambda\in\Lambda_1} H_\lambda = \cl{\conv{\mc F}} = \conv{\mc F} = \mc F.
\end{align} 
Finally we show that $\mc F = \im P_c$. Note that $\mc F\subset \im P_c$ by definition. Corollary~\ref{corollary:intersection formula for F} proves that $\cl{\conv{\im P_c}} = \cl{\conv{\mc F}}$. 
We therefore have
\begin{align*}
\cl{\conv{\im P_c}} = \cl{\conv{\mc F}} = \mc F \subset \im P_c.
\end{align*}
Since $\im P_c \subset \cl{\conv{\im P_c}}$, we have $\mc F = \im P_c$.
\end{proof}
\begin{remark}
Theorem~\ref{theorem:convexity of F} shows that the image of $P_c$ coincides with $\mc F$. This means that, if nonpositive voltage potentials would be permitted, then any feasible power demand that is satisfied by nonpositive voltage potentials can also be satisfied by positive voltage potentials. Hence, from a theoretical standpoint, the restriction to positive voltage potentials does not make the set of feasible power demands more conservative.
\end{remark}

Due to the convexity of $\mc F$, Corollary~\ref{corollary:boundary D F one-to-one} implies that $\partial \mc D$ and $\partial \mc F$ in are one-to-one correspondence. Moreover, Lemma~\ref{lemma:lambda f} implies that there are no vectors $x\not\in\partial \mc D$ such that $P_c(x)\in\partial\mc F$. This implies the following corollary.

\begin{corollary}[\mainresult{one-to-one boundary}]\label{corollary:one-to-one boundary}
For each $\tilde P_c$ on the boundary of $\mc F$ there exist a unique $\tilde V_L\in\RR^n$ %
that satisfies \eqref{eqn:dc power flow equation}. All such $\tilde V_L$ satisfy $\tilde V_L>\mb 0$ and form the boundary of $\mc D$. 
Hence, there is a one-to-one correspondence between $\partial \mc D$ and $\partial \mc F$.
\end{corollary}

Theorem~\ref{theorem:one-to-one correspondence} and Corollary~\ref{corollary:one-to-one boundary} immediately imply that there is a one-to-one correspondence between the set $\mc D$ of long-term voltage stable operating points and the power demands which are feasible under small perturbation, by which we mean that such a power demand $\tilde P_c$ is feasible and does not lie on the boundary of $\mc F$ (\ie, $\tilde P_c \in \inter{\mc F}$).
Consequently, if a power demand is feasible under small perturbation, then there exists a unique long-term voltage stable operating point which satisfies the power flow equation.
\begin{corollary}
There is a one-to-one correspondence between the long-term voltage stable operating points $\mc D$ and the feasible power demands under small perturbations $\inter{\mc F}$.
\end{corollary}

\subsection{A necessary and sufficient LMI condition for feasibility}\label{subsection:necessary and sufficient condition}
We conclude Part~I of this paper by restating the geometric characterization of $\mc F$ in Theorem~\ref{theorem:convexity of F} in terms of an LMI condition.
In the context of Problem~\ref{problem:dc power flow problem}, \cite{barabanov2016} presents a necessary LMI condition for the feasibility of power demands, and states that the LMI condition is also necessary when the set of feasible power demands is closed and convex, as is the case here.
The next theorem recovers this result and extends the result for power demands which are feasible under small perturbation. %
\begin{theorem}[\mainresult{necessary and sufficient condition}]\label{theorem:necessary and sufficient condition}
A vector $\tilde P_c$ of power demands is feasible (\ie, $\tilde P_c\in\mc F$) if and only if there does not exists a positive vector $\nu\in\RR^n$ such that the $(n+1) \times (n+1)$ matrix
\begin{align}\label{eqn:necessary and sufficient condition}
\begin{pmatrix}
[\nu] Y_{LL} + Y_{LL}[\nu] & [\nu] \mc I_L^* \\ ([\nu] \mc I_L^*)\T & 2 \nu\T \tilde P_c
\end{pmatrix} = 2 \begin{pmatrix}
h(\nu) & \tfrac 1 2 [\nu] \mc I_L^* \\ \tfrac 1 2 ([\nu] \mc I_L^*)\T & \nu\T \tilde P_c
\end{pmatrix}
\end{align}
is positive definite.
Similarly, $\tilde P_c$ is feasible under small perturbation (\ie, $\tilde P_c\in\inter{\mc F}$) if and only if there does not exists a positive vector $\nu\in\RR^n$ such that \eqref{eqn:necessary and sufficient condition} is positive semi-definite.
\end{theorem}
\begin{proof}
We will prove the logical transposition.

($\Leftarrow$): Without loss of generality we assume that $\|\nu\|_1 = 1$. If \eqref{eqn:necessary and sufficient condition} is positive semi-definite, then $h(\nu)$ is positive semi-definite. It follows from Lemmas~\ref{lemma:closure of Lambda} and \ref{lemma:Lambda properties} that $h(\nu)$ is an irreducible M-matrix. Let $v>\mb 0$ be a Perron vector of $h(\nu)$. Suppose that $h(\nu)$ is singular, then $h(\nu)v=\mb 0$ by Proposition~\ref{proposition:perron root M-matrix}. However, note that for $t\in\RR$ we have
\begin{align*}
\begin{pmatrix}
t v \\ 1
\end{pmatrix}\T\begin{pmatrix}
h(\nu) & \tfrac 1 2 [\nu] \mc I_L^* \\ \tfrac 1 2 ([\nu] \mc I_L^*)\T & \nu\T \tilde P_c
\end{pmatrix}
\begin{pmatrix}
t v \\ 1
\end{pmatrix} = t v\T [\nu]\mc I_L^* + \nu\T \tilde P_c,
\end{align*}
which is a nonconstant line in $t$ since $v\T [\nu]\mc I_L^*>0$, and is not bounded from below.
This contradicts the assumption that \eqref{eqn:necessary and sufficient condition} is positive semi-definite. Hence $h(\nu)$ must be positive definite and $\nu\in\Lambda_1$. Alternatively, if \eqref{eqn:necessary and sufficient condition} is positive definite, then $h(\nu)$ is positive definite. If $h(\nu)$ is positive definite, then by the Haynsworth inertia additivity formula (\cite{zhang2006schur}, Sec. 0.10) \eqref{eqn:necessary and sufficient condition} is positive definite (semi-definite) 
if and only if 
\begin{align}\label{eqn:necessary and sufficient condition:1}
\nu\T \tilde P_c - \tfrac 1 4 ([\nu] \mc I_L^*)\T h(\nu)\inv [\nu] \mc I_L^* > (\ge)\: 0.
\end{align}
Using \eqref{eqn:definition of m} and \eqref{eqn:definition norm}, we note that \eqref{eqn:necessary and sufficient condition:1} is equivalent to
\begin{align}\label{eqn:necessary and sufficient condition:2}
\nu\T \tilde P_c > (\ge)\: \tfrac 1 4 ([\nu] \mc I_L^*)\T h(\nu)\inv [\nu] \mc I_L^* = \|\varphi(\nu)\|_{h(\nu)}^2.
\end{align}
Theorem~\ref{theorem:convexity of F} implies that $\tilde P_c$ is not feasible if and only if there exists $\lambda\in\Lambda$ such that $\tilde P_c\not\in\ H_{\lambda}$, or equivalently, $\lambda\T \tilde P_c > \|\varphi(\lambda)\|_{\lambda}^2$. Thus, if \eqref{eqn:necessary and sufficient condition} is positive definite, then the strict inequality in \eqref{eqn:necessary and sufficient condition:2} holds and $\tilde P_c$ is not feasible. 
Moreover, if equality in \eqref{eqn:necessary and sufficient condition:2} holds then 
\begin{align*}
\nu\T \tilde P_c = \|\varphi(\nu)\|_{h(\nu)}^2 = \nu\T P_c(\varphi(\nu)).
\end{align*}
Lemma~\ref{lemma:lambda f} implies that $\tilde P_c = P_c(\varphi(\nu))$, and thus $\tilde P_c\in\partial \mc F$ by Theorem~\ref{theorem:boundary of F}. Thus, if \eqref{eqn:necessary and sufficient condition} is positive semi-definite, then $\tilde P_c \not\in \mc F$ or $\tilde P_c \in \partial\mc F$, and therefore $\tilde P_c\not\in\inter{\mc F}$.

($\Rightarrow$): The converse is obtained by reversing the steps.
\end{proof}
Theorem~\ref{theorem:necessary and sufficient condition} presents a necessary and sufficient LMI conditions for the feasibility (under small perturbation) of a DC power grid with constant-power loads.
A more common formulation of Theorem~\ref{theorem:necessary and sufficient condition} as an LMI condition can be obtained by replacing $[\nu]$ by a positive definite diagonal matrix $D$, and replacing $\nu\T \tilde P_c$ by $\mb 1\T D \tilde P_c$ (\textit{cf.} \cite{barabanov2016}).
\section{Conclusion of Part I}\label{section:conclusion}
In Part I of this paper we have studied the power flow feasibility of DC power grids with constant-power loads, and have presented a framework for the analysis of this feasibility problem.
Specifically, we have presented a geometric characterization of the feasible power demands in terms of half-spaces, along with necessary and sufficient LMI conditions to check if a vector of power demands is feasible (under small perturbation). 
In addition, we have given a novel proof for the convexity of the set of feasible power demands.
More importantly, we proved that there exists a one-to-one correspondence between the feasible power demands and the long-term voltage semi-stable operating points. This shows that for each feasible power demand there exists a unique operating point which is long-term voltage semi-stable and satisfies the power flow equations. 
This operating point can be found by solving an initial value problem.
The existence and uniqueness of this operating point proves that long-term (semi-)stability can be guaranteed for each feasible power demand. 
Furthermore, we showed that there exists a one-to-one correspondence between the feasible power demands under small perturbations and the long-term voltage stable operating points.

Our analysis is continued in Part II of this paper, in which we study high-voltage operating points and sufficient conditions for power flow feasibility, among other things.

\appendix
\renewcommand{\theproposition}{A.\arabic{theorem}}
\renewcommand{\thelemma}{A.\arabic{theorem}}
\subsection{Properties of Z-, M- and irreducible matrices}\label{subsection:matrix theory}

\begin{proposition}\label{proposition:continuity of perron root and vector}
The Perron root and Perron vector of an irreducible Z-matrix $A$ are continuous in the elements of $A$.
\end{proposition}
\begin{proof}
All eigenvalues and any eigenvector corresponding to a simple eigenvalue are continuous in the elements of the corresponding matrix (see \cite[3.1.2]{ortega1990numerical} and \cite[3.1.3]{ortega1990numerical}).
\end{proof}

\begin{proposition}[{\cite[Thm. 5.8]{fiedler1986special}}]\label{proposition:perron root M-matrix}
An irreducible M-matrix is singular if and only if its Perron root is zero, in which case its kernel is spanned by any Perron vector.
\end{proposition}

\begin{proposition}\label{proposition:diagonal properties}
Consider the diagonal matrix $[d]$, $d\in\RR^n$. The following statements hold:
\begin{enumerate}
\item If $A$ is irreducible, then $A+[d]$ is irreducible; \label{proposition:diagonal properties:irreducible addition}
\item If $A$ is a Z-matrix, then $A+[d]$ is a Z-matrix; \label{proposition:diagonal properties:Z-matrix addition}
\item If $A$ is irreducible and $d>\mb 0$, then $A[d]$ are $[d]A$ are irreducible; \label{proposition:diagonal properties:irreducible multiplication}
\item If $A$ is a Z-matrix and $d>\mb 0$, then $A[d]$ are $[d]A$ are Z-matrices; \label{proposition:diagonal properties:Z-matrix multiplication}
\item If $A$ is an M-matrix and $d>\mb 0$, then $A[d]$ are $[d]A$ are M-matrices; \label{proposition:diagonal properties:M-matrix multiplication}
\item If $A$ is an irreducible M-matrix and $d\gneqq\mb 0$, then $A+[d]$ is a nonsingular irreducible M-matrices; \label{proposition:diagonal properties:irreducible M-matrix addition}
\end{enumerate}
\end{proposition}
\begin{proof}
Statements \ref{proposition:diagonal properties:irreducible addition}, \ref{proposition:diagonal properties:Z-matrix addition}, \ref{proposition:diagonal properties:irreducible multiplication} and \ref{proposition:diagonal properties:Z-matrix multiplication} follow directly from the definitions of a Z-matrix and irreducible matrix.
Statement \ref{proposition:diagonal properties:M-matrix multiplication} is shown in \cite{PLEMMONS1977175}.
Statement \ref{proposition:diagonal properties:irreducible M-matrix addition} follows from \cite[Cor.~3.9]{li1995characterizations} if $A$ is singular, and is implied by \cite[Thm. 1, {$\mathrm{A}_3$}]{PLEMMONS1977175} if $A$ is nonsingular.
\end{proof}
\begin{proposition}\label{proposition:sum of irreducible Z-matrices}
The sum of two irreducible Z-matrices is an irreducible Z-matrix.\label{proposition:diagonal properties:sum of Z-matrix multiplication}
\end{proposition}
\begin{proof}
If $A_{[\alpha,\alpha\comp]}\lneqq 0$ and $B_{[\alpha,\alpha\comp]}\lneqq 0$ then $(A+B)_{[\alpha,\alpha\comp]}\lneqq 0$ for all nonempty $\alpha\subsetneqq \boldsymbol n$.
\end{proof}

\renewcommand{\thelemma}{B.\arabic{theorem}}
\subsection{Properties of $\Lambda_1$ and its closure}\label{subsection:closure of Lambda}

Recall from \eqref{eqn:Lambda} and \eqref{eqn:Lambda_1} that $\Lambda$ ($\Lambda_1$) is the set of vectors $\lambda$ such that $h(\lambda)=\tfrac 1 2 ([\lambda]Y_{LL} + Y_{LL}[\lambda])$ is positive definite (and ${\|\lambda\|_1=1}$).

\begin{lemma}\label{lemma:Lambda convex}
The set $\Lambda$ is an open convex cone.
\end{lemma}
\begin{proof}
The convex combination of positive definite matrices is again positive definite, and the set of all positive definite matrices is open. The result follows since $h(\lambda)$ is linear in $\lambda$.
\end{proof}

\begin{lemma}\label{lemma:Lambda positive}
The set $\Lambda$ is contained in the positive orthant. \Ie, $\lambda>\mb 0$ for $\lambda\in\Lambda$. 
\end{lemma}
\begin{proof}
Let $\lambda$ be such that $h(\lambda)$ is positive definite.
Recall that the matrix $Y_{LL}$ is positive definite.
A matrix is positive definite only if its diagonal elements are positive. 
The diagonal elements of $Y_{LL}$ and $h(\lambda)$ are respectively given by $(Y_{LL})_{ii}$ and $\lambda_i (Y_{LL})_{ii}$, and therefore
$(Y_{LL})_{ii}>0$ and $\lambda_i (Y_{LL})_{ii}>0$. This implies that $\lambda_i>0$ for all $i$.
\end{proof}

\begin{lemma}
The set $\Lambda_1$ is a bounded convex set.
\end{lemma}
\begin{proof}
Let $\lambda\in\Lambda$.
Since $\lambda>\mb 0$ by Lemma~\ref{lemma:Lambda positive}, it follows that $\|\lambda\|_1 = \lambda\T \mb 1$.
Hence $\Lambda_1 = \Lambda \cap \set{\lambda}{\|\lambda\|_1 = 1} = \Lambda \cap \set{\lambda}{\lambda\T \mb 1 = 1}$.
The latter expression is an intersection of convex sets (see Lemma~\ref{lemma:Lambda convex}). Hence $\Lambda_1$ is convex.
The set $\set{\lambda}{\|\lambda\|_1 = 1}$ is bounded and thus $\Lambda_1$ is bounded.
\end{proof}

\begin{lemma}\label{lemma:closure of Lambda}
The closure of $\Lambda_1$ satisfies
\begin{align*}
\cl{\Lambda_1} = \set{\lambda}{h(\lambda)\text{ is positive semi-definite},\|\lambda\|_1=1 }.
\end{align*}
\end{lemma}
\begin{proof}
Since $\Lambda_1$ is nonempty, this follows directly from linearity of $h$, and the fact that the positive semi-definite matrices form the closure of the positive definite matrices.%
\end{proof}

\begin{lemma}\label{lemma:Lambda properties}
The set $\cl{\Lambda_1}$ is contained in the positive orthant. Moreover, the matrix $h(\lambda)$ for $\lambda\in\cl{\Lambda_1}$ is an irreducible M-matrix.
\end{lemma}

\begin{proof}
The vectors in $\Lambda_1$ are positive, and so the vectors in $\cl{\Lambda_1}$ are nonnegative.
To show that $\cl{\Lambda_1}$ lies in the positive orthant, it suffices to show that if a vector ${\lambda\in\partial \Lambda_1}$ contains zeros, then $Y_{LL}$ is not irreducible, which is a contradiction.

Suppose $\lambda\in \partial \Lambda_1$ such that $\lambda_{[\alpha]}=\mb 0$ and $\lambda_{[\alpha\comp]} > \mb 0$ for some nonempty set $\alpha\subset\boldsymbol n$.
Let $w$ be a vector such that $w_{[\alpha\comp]} = \mb 0$ and $w_{[\alpha]}$ is arbitrary. We therefore have $[\lambda]w = \mb 0$.
Since $h(\lambda)$ is positive semi-definite, the following inequality holds for every vector $v$ and scalar $\beta$:
\begin{multline}\label{eqn:Lambda properties:inequality}
\hspace{-1em} 0\le (v-\beta w)\T h(\lambda)(v-\beta w) 
= (v-\beta w)\T [\lambda]Y_{LL} (v-\beta w) \\= v\T [\lambda]Y_{LL} (v-\beta w)= v\T[\lambda]Y_{LL}v - \beta v\T[\lambda]Y_{LL}w
\end{multline}
If $v$ is such that $ v\T[\lambda]Y_{LL}w\neq 0$, then \eqref{eqn:Lambda properties:inequality} is violated when we take $\beta$ such that $\beta v\T[\lambda]Y_{LL}w$ is sufficiently large.
It follows that $v\T[\lambda]Y_{LL}w = 0$ for all $v$.
This implies that 
\begin{align}\label{eqn:Lambda properties:Perron vector identity}
[\lambda]Y_{LL}w=\mb 0.
\end{align}
The rows of \eqref{eqn:Lambda properties:Perron vector identity} corresponding to $\alpha\comp\subset\boldsymbol n$ satisfy
\begin{align}\label{eqn:Lambda properties:alpha comp rows}
[\lambda_{[\alpha\comp]}] ((Y_{LL})_{[\alpha\comp,\alpha]}w_{[\alpha]} + (Y_{LL})_{[\alpha\comp,\alpha\comp]}w_{[\alpha\comp]}) = \mb 0.
\end{align}
Recall that $\lambda_{[\alpha\comp]}>\mb 0$ and $w_{[\alpha\comp]} = \mb 0$, and thus \eqref{eqn:Lambda properties:alpha comp rows} implies
\begin{align}\label{eqn:Lambda properties:contradiction}
(Y_{LL})_{[\alpha\comp,\alpha]}w_{[\alpha]} = \mb 0.
\end{align}
Since $w_{[\alpha]}$ is arbitrary, \eqref{eqn:Lambda properties:contradiction} should hold for all $w_{[\alpha]}$, and hence $(Y_{LL})_{[\alpha\comp,\alpha]}=0$. However, this contradicts the assumption that $Y_{LL}$ is irreducible. We conclude that $\lambda>\mb 0$.
Since $Y_{LL}$ is an irreducible Z-matrix and $\lambda > \mb 0$, Propositions~\ref{proposition:diagonal properties} and \ref{proposition:sum of irreducible Z-matrices} imply that $h(\lambda)$ is an irreducible Z-matrix. Since $h(\lambda)$ is positive semi-definite, its eigenvalues are real and nonnegative, and so $h(\lambda)$ is an M-matrix.
\end{proof}

\begin{lemma}\label{lemma:positive semi-definite h}
Let $\lambda\in\partial\Lambda_1$.
For every $v\gneqq\mb 0$ in the kernel of $h(\lambda)$ we have $\lambda\T \pdd {P_c} {V_L} (x)v > 0$ for all $x\in\RR^n$.
\end{lemma}
\begin{proof}
Since $\lambda\in\partial\Lambda_1$, the matrix $h(\lambda)$ is singular.
By Lemma~\ref{lemma:Lambda properties}, it is a singular irreducible M-matrix, and its Perron root is zero.
The kernel of $h(\lambda)$ is spanned by any Perron vector, by Proposition~\ref{proposition:perron root M-matrix}.
This implies that $v$ is a Perron vector and $v>\mb 0$. 
Let $x$ be any vector. By substituting \eqref{eqn:jacobian of f}, we note that $\lambda \T \pdd {P_c} {V_L} (x) v$ is equivalent to
\begin{multline}\label{eqn:halfspace image:contradiction formula}
\lambda \T \pdd {P_c} {V_L} (x) v = \lambda\T [Y_{LL}V_L^*] v - \lambda \T ([x]Y_{LL}+[Y_{LL}x]) v \\
= \lambda\T [\mc I_L^*] v - 2x\T h(\lambda) v = \lambda\T [\mc I_L^*] v,
\end{multline}
where we used the fact that $h(\lambda) v = \mb 0$.
Since $v>\mb 0$, $\lambda>\mb 0$ and $\mc I_L^*\gneqq \mb 0$, it follows from \eqref{eqn:halfspace image:contradiction formula} that $\lambda\T [\mc I_L^*] v > 0$. 
We conclude that $\lambda \T \pdd {P_c} {V_L} (x) v > 0$. 
\end{proof}
\subsection{Proofs concerning Section~\ref{subsection:convex hull of F}}\label{appendix:proof of convex hull of F}
\subsubsection*{Proof of Theorem~\ref{theorem:supporting half-spaces of F}}
($\Leftarrow$): 
The half-space $H(\lambda,s)$ with $\lambda\in \Lambda_1$ and $s=\|\varphi (\lambda)\|_{h(\lambda)}^2$ is given by
\begin{align*}
H(\lambda,s) = \set{y}{\lambda\T y \le \|\varphi (\lambda)\|_{h(\lambda)}^2}.
\end{align*}
Since $\lambda\in\Lambda_1$, Lemma~\ref{lemma:lambda f} states that \eqref{eqn:lambda f inequality} holds for all $x\in\RR^n$. %
This implies that $P_c(x)\in H(\lambda,s)$ for all $x\in\RR^n$, and thus $\im P_c \subset H(\lambda,s)$.
To show that $P_c(\varphi (\lambda))$ is the unique point of support, we show that 
\begin{align}\label{eqn:halfspace image:support}
\cl{\im P_c} \cap \partial H(\lambda,s) = \{P_c(\varphi (\lambda))\}.
\end{align}
Let $y\in \cl{\im P_c} \cap \partial H(\lambda,s)$, then there exists a sequence $\{x_k\}_{k\in\NN}\in\RR^n$ such that
\begin{align}\label{eqn:halfspace image:limit 2}
\lim_{k\to\infty}P_c(x_k)= y.
\end{align} %
Since $y\in\partial H(\lambda,s)$, multiplying \eqref{eqn:halfspace image:limit 2} by $\lambda\T $ yields
\begin{align}\label{eqn:halfspace image:limit}
\lim_{k\to\infty}\lambda\T P_c(x_k)= \lambda\T y = s = \|\varphi (\lambda)\|_{h(\lambda)}^2.
\end{align}
It follows from rearranging \eqref{eqn:halfspace image:limit}  and applying \eqref{eqn:lambda f} that
\begin{multline*}
0 = \lim_{k\to\infty} \l(\|\varphi (\lambda)\|_{h(\lambda)}^2 - \lambda\T P_c(x_k)\r) \\= \lim_{k\to\infty}\|\varphi (\lambda) - x_k\|_{h(\lambda)}^2.
\end{multline*} Hence $\lim_{k\to\infty} x_k = \varphi (\lambda)$, and so \eqref{eqn:halfspace image:support} holds.
This proves that $H(\lambda,s)$ supports $\im P_c$, and that $P_c(\varphi (\lambda))$ is a point of support.
The same is true for $\mc F$ since $\mc F\subset \im P_c$ and $P_c(\varphi (\lambda))\in\mc F$.

($\Rightarrow$): 
Let $x>\mb 0$ be a vector. Let $\lambda$ be such that $\|\lambda\|_1=1$ and $\lambda\not\in\Lambda_1$, which means that $h(\lambda)$ is not positive definite.
We will show that there exists a vector $v\ge \mb 0$ such that $\lambda\T P_c(x + t v)$ for scalars $t\ge 0$ is not bounded from above. Since $x+tv>\mb 0$ for all $t\ge 0$, this implies %
that the hyperplane $H(\lambda,s)$ does not contain $\mc F$ for any scalar $s$. The same holds for $\im P_c$ since $\mc F\subset \im P_c$.
Lemma~\ref{lemma:sum formula f} yields
\begin{align}\label{eqn:halfspace image:line mapping}
P_c(\hat x + t v) = P_c(\hat x)+ t \pdd {P_c} {V_L} (\hat x)v - t^2[v]Y_{LL}v.
\end{align}
We multiply \eqref{eqn:halfspace image:line mapping} by $\lambda\T$ and use \eqref{eqn:rewriting quadratic}, which implies
\begin{align}\label{eqn:halfspace image:function sum}
\lambda\T P_c(\hat x + t v) = \lambda\T P_c(\hat x) + t \lambda \T \pdd {P_c} {V_L} (\hat x) v - t^2 v\T h(\lambda)v.
\end{align}

If $\lambda_i<0$ for some $i$, then $e_i\T h(\lambda) e_i = (Y_{LL})_{ii} \lambda_i <0$. Hence, taking $v=e_i\gneqq \mb 0$ in \eqref{eqn:halfspace image:function sum} describes a parabola in $t$ which is not bounded from above. Thus $\lambda\T P_c(x)$ is not bounded from above for $t\ge 0$.

If $\lambda\ge \mb 0$ and $\lambda\not\in \cl{\Lambda_1}$ then the matrix $h(\lambda)$ has a negative eigenvalue by Lemma~\ref{lemma:closure of Lambda}.
Let $r$ be the eigenvalue of $h(\lambda)$ with the smallest (\ie, most negative) real part.
Since $\lambda\ge \mb 0$, it follows that $h(\lambda)$ is a Z-matrix.
The matrix $h(\lambda)$ is block diagonal, where each block corresponds to an irreducible component of $h(\lambda)$.
Let $h(\lambda)_{[\alpha,\alpha]}$ be the irreducible component that corresponds to the negative eigenvalue $r$.
The matrix $h(\lambda)_{[\alpha,\alpha]}$ is an irreducible Z-matrix with Perron root $r$ and Perron vector $w>\mb 0$.
Let $v$ in \eqref{eqn:halfspace image:function sum} be such that $v_{[\alpha]} = w$ and $v_{[\alpha\comp]}=\mb 0$, then $v\T h(\lambda)v= r w\T w <0$. It follows that \eqref{eqn:halfspace image:function sum} describes a parabola in $t$ which is not bounded from above. Thus $\lambda\T P_c(x)$ is not bounded from above for $t\ge 0$.

Finally, suppose $\lambda\in\partial \Lambda_1$, which implies by Lemma~\ref{lemma:Lambda properties} that $\lambda>\mb 0$ and that $h(\lambda)$ is an irreducible M-matrix.
The matrix $h(\lambda)$ is singular since $\lambda\not\in\Lambda_1$.
Let $v>\mb 0$ in \eqref{eqn:halfspace image:function sum} be a Perron vector of $h(\lambda)$.
Proposition~\ref{proposition:perron root M-matrix} states that $v$ spans the kernel of $h(\lambda)$, and so $v\T h(\lambda)v=0$. 
By Lemma~\ref{lemma:positive semi-definite h} we know that $\lambda \T \pdd {P_c} {V_L} (x) v > 0$.
This implies that \eqref{eqn:halfspace image:function sum} describes a half-line for $t\ge 0$ which is not bounded from above.
Hence, $\lambda\T P_c(x)$ is not bounded from above for $t\ge 0$. %
\hfill\QED

\subsubsection*{Proof of Theorem~\ref{theorem:boundary of F}}\label{appendix:proof of boundary of F}
The half-spaces $H_\lambda$ for $\lambda\in\Lambda_1$ are all supporting half-spaces of $\conv{\mc F}$, which follows from \eqref{eqn:intersection formula for F} of Corollary~\ref{corollary:intersection formula for F}. 
Theorem~\ref{theorem:supporting half-spaces of F} proves that $P_c(\varphi (\lambda))$ is a point of support to $H_\lambda$, and that it is unique in the case of $\mc F$. 
Theorem 2.15 of \cite{valentine1964convex} states that all boundary points of a convex set are a point of support associated to some supporting half-space. 
This implies that $P_c(\partial\mc D) \subset \partial\conv{\mc F}$.
We prove equality by showing that there are no other points of support.

Let $y\in \cl{\conv{\mc F}}$. Then there exists a sequence $\{y_k\}_{k\in\NN}\in\conv{\mc F}$ such that $\lim_{k\to\infty}y_k= y$.
This means that for $k\in\NN$ there exists $x_k,z_k\in\RR^n$ and scalars $\theta_k$ such that $y_k = \theta_k P_c(x_k)+(1-\theta_k)P_c(z_k)$ and $0<\theta_k<1$.
Suppose there exists $\widetilde \lambda\in\Lambda_1$ so that $y\in\partial H_{\widetilde \lambda}$. Hence, $y$ is a point of support associated to $H_{\widetilde \lambda}$.
We define $s:=\|\varphi (\lambda)\|_{h(\lambda)}^2$ and observe that
\begin{multline*}
\lim_{k\to\infty}\widetilde \lambda\T(\theta_k P_c(x_k)+(1-\theta_k)P_c(z_k)) \\= \lim_{k\to\infty}\widetilde \lambda\T y_k = \widetilde \lambda\T y = s.
\end{multline*}
Lemma~\ref{lemma:lambda f} implies that for all $k$ we have $\widetilde \lambda\T P_c(x_k)\le s$, with equality if and only if $x_k = \varphi (\widetilde \lambda)$, and the same holds for $z_k$. This implies that
\begin{align}\label{eqn:boundary of F:inequality}
\widetilde \lambda\T(\theta_k P_c(x_k)+(1-\theta_k)P_c(z_k)) \le s. 
\end{align}
In order to converge to equality in \eqref{eqn:boundary of F:inequality} as $k\to\infty$, we require that either $x_k\to \varphi (\widetilde \lambda)$ and $z_k\to \varphi (\widetilde \lambda)$, $x_k \to \varphi (\widetilde \lambda)$ and $\theta_k \to 1$, or $z_k \to \varphi (\widetilde \lambda)$ and $\theta_k\to 0$. In all cases it follows that $y = \lim_{k\to\infty} y_k = P_c(\varphi (\widetilde \lambda))$. Hence $P_c(\varphi (\widetilde \lambda))$ is the unique point of support.

Lemma~\ref{lemma:lambda f} implies that if $P_c(\varphi (\lambda_1)) = P_c(\varphi (\lambda_2))$, then $\varphi (\lambda_1) = \varphi (\lambda_2)$. Hence the map $\lambda\to P_c(\varphi (\lambda))$ is a one-to-one correspondence between $\Lambda_1$ and $\partial\conv{\mc F}$, and $\lambda\to P_c(\varphi (\lambda))$ for $\lambda\in\Lambda_1$ parametrizes $\partial\conv{\mc F}$.

Note the inclusion 
\begin{align*}
\partial \conv{\mc F} = P_c(\partial \mc D)\subset \mc F\subset \conv{\mc F},
\end{align*} which implies that $\conv{\mc F}$ is closed.\hfill\QED
\subsection{Proofs concerning Section~\ref{subsection:convexity of F}}\label{appendix:proofs for convexity of F}
\subsubsection*{Proof of Lemma~\ref{lemma:convex combination}}
First we show that if a path $\gamma:[0,T]\to\mc D$ satisfies \eqref{eqn:convex combination:differential equation} with $\gamma(0)=\hat V_L\in \mc D$, then \eqref{eqn:convex combination:convex combination} holds.
Indeed, note that the matrix $\pdd {P_c} {V_L} (\gamma(\tau))$ is invertible for $0\le\tau \le T$ since $\gamma(\tau)\in\mc D$, and note for $0\le\theta\le T$ that by the fundamental theorem of calculus we have
\begin{align}\label{eqn:convex combination:fundamental theorem of calculus}
P_c (\gamma(\theta)) &= P_c (\gamma(0)) + \int_0^\theta \pdd {P_c} {V_L} (\gamma(\tau))\dot \gamma(\tau) \d\tau.
\end{align}
Substitution of \eqref{eqn:convex combination:differential equation} in \eqref{eqn:convex combination:fundamental theorem of calculus} yields
\begin{align}\label{eqn:convex combination:necessary consequence}
P_c (\gamma(\theta)) &= P_c (\gamma(0)) + \int_0^\theta (\tilde P_c - \hat P_c) \d\tau.
\end{align}
Eq. \eqref{eqn:convex combination:convex combination} follows from \eqref{eqn:convex combination:necessary consequence} since $P_c (\gamma(0)) = P_c (\hat V_L) = \hat P_c$.

To complete the proof it remains to show that a solution $\gamma(\theta)\in \mc D$ to \eqref{eqn:convex combination:differential equation} for $\theta\in[0,1]$ exists and that this solution is unique. %
Let the map $\psi: \mc D\to\RR^n$ be defined by
\begin{align*}
\psi(z) := \l(\pdd {P_c} {V_L} (z)\r)\inv (\tilde P_c-\hat P_c).
\end{align*}
The map $\psi(z)$ is continuously differentiable since $\pdd {P_c} {V_L} (z)$ is invertible for $z\in\mc D$.
Corollary 8.17 of \cite{kelley2010theory} states that the initial value problem \eqref{eqn:convex combination:differential equation} %
has a unique solution $\gamma:(-\varepsilon,\varepsilon)\to B$ for some $\varepsilon > 0$, where $B$ is an open neighborhood of $\hat V_L$ which is contained in $\mc D$.
Since $\psi(z)$ is continuous at all $z\in\mc D$, the solution $\gamma$ can be extended to a maximal interval of existence.
Indeed, by the Theorem 8.33 of \cite{kelley2010theory} we extend $\gamma$ so that either \mbox{(i) $\gamma(\theta) \to \partial \mc D$,} or (ii) $|\gamma(\theta)_i| \to \infty$ for some $i\in\boldsymbol n$, as $\theta\to \omega$ where $\omega\in\RR_{>0} \cup\{+\infty\}$. We will treat cases (i) and (ii) separately.
\subsubsection*{Case (i)}
Let $x\in\partial D$ such that $\gamma(\theta) \to x$ as $\theta\to \omega$ and let $y:= P_c(x)$. By continuity of $P_c(V_L)$ it follows that $P_c(\gamma(\theta)) \to y$ as $\theta\to \omega$. Since $\gamma(\theta)\in\mc D$ for $0\le \theta<\omega$, the first part of this proof showed that \eqref{eqn:convex combination:convex combination} holds for $0\le \theta<\omega$. Suppose $\omega = 1$, then taking the limit $\theta \to \omega$ in \eqref{eqn:convex combination:convex combination} implies that $\tilde P_c = P_c(\gamma(\omega)) = y$, which lies on the boundary of $\mc F$. This contradicts the fact that $\tilde P_c\in\inter{\conv{\mc F}}$.
Suppose $1 > \omega$, then \eqref{eqn:convex combination:convex combination} implies that $P_c(\gamma(\omega)) = y$ is a convex combination of $\hat P_c$ and $\tilde P_c$. Since $\hat V_L\not\in\partial\mc D$ it follows from  Theorem~\ref{theorem:boundary of F} that $\hat P_c\not\in \partial\conv{\mc F}$. Since $\hat P_c\in \mc F$ we therefore have $\hat P_c\in\inter{\conv{\mc F}}$. Let $\lambda\in\Lambda_1$ such that $P_c(\varphi(\lambda))=y$, which exists by Theorem~\ref{theorem:boundary of F}, and define $s:=\|\varphi(\lambda)\|_{h(\lambda)}^2 = \lambda\T P_c(\varphi(\lambda)) =\lambda\T y$. Note that $\lambda\T \hat P_c < s$ and $\lambda\T \tilde P_c < s$ since $\hat P_c,\tilde P_c\in\inter{\conv{\mc F}}$. But since $y$ is a convex combination of $\hat P_c$ and $\tilde P_c$, this would imply that $\lambda\T y < s$, which is a contradiction.
We conclude that $1 < \omega$, and in particular $\gamma(\theta)\in\mc D$ for $0\le \theta \le 1$.
\subsubsection*{Case (ii)}
We will show that $\omega=+\infty$ and that $P_c(\gamma(\theta))$ describes a half-line for $0\le\theta<\infty$.
Let $\lambda\in\Lambda_1$. Note that $|\gamma(\theta)_i|\to\infty$ implies that also $|\varphi(\lambda)_i - \gamma(\theta)_i| \to \infty$. Therefore also $\|\varphi(\lambda)-\gamma(\theta)\|_{h(\lambda)} \to \infty$. It follows from Lemma~\ref{lemma:lambda f} that $\lambda\T P_c(\gamma(\theta))\to-\infty$. This holds for all $\lambda\in\Lambda_1$ and so $P_c(\gamma(\theta))$ does not intersect the boundary of $\mc F$ for $0\le\theta<\omega$. The first part of this proof showed that \eqref{eqn:convex combination:convex combination} holds for $0\le \theta<\omega$, which describes a half-line in $\theta$. Since $\lambda>\mb 0$ by Lemma~\ref{lemma:Lambda positive}, it follows from $\lambda\T P_c(\gamma(\theta))\to-\infty$ that $P_c(\gamma(\theta))_j\to-\infty$ for some $j\in \boldsymbol n$. As a result, \eqref{eqn:convex combination:convex combination} implies that $\omega = +\infty$.
In particular it follows that $\tilde P_c$ lies on the half-line and that $\gamma(\theta)\in\mc D$ for $0\le \theta \le 1$.

To show uniqueness, we note again that $\psi(z)$ is continuously differentiable. Corollary 8.17 of \cite{kelley2010theory} states that %
\eqref{eqn:convex combination:differential equation}
has a unique solution in an open neighborhood around \textit{any} given initial value in $\mc D$. 
Taking any point $\gamma(\theta)$ with $0\le\theta\le 1$ as an initial value shows that the solution $\gamma$ is unique at each point, and hence is unique in $\mc D$.

Since \eqref{eqn:convex combination:convex combination} holds for $0 \le \theta \le 1$, \eqref{eqn:convex combination:convex combination} implies that $\tilde P_c = P_c(\gamma(1))\in P_c(\mc D)$.
\hfill\QED

\end{document}